\documentclass[11pt]{article}

\usepackage{indentfirst}
\usepackage{amssymb}
\usepackage[utf8]{inputenc}
\usepackage{mathptm}
\usepackage{ae}
\usepackage{bm}
\usepackage{amsthm,amsfonts,latexsym,fancyhdr,amsmath,tabularx}
\usepackage{graphicx}
\usepackage{subfigure, multicol}
\usepackage{epsfig}
\usepackage{verbatim}
\usepackage{pgf,tikz,pgfplots,gnuplottex}
\usepackage{multirow}
\usepackage{url,fancybox,color,makeidx}
\pgfplotsset{compat=1.17}
\usetikzlibrary{arrows}

\newtheorem{theorem}{Theorem}[section]

\newtheorem{corollary}{Corollary}[section]
\newtheorem{definition}{Definition}[section]

\newtheorem{remark}{Remark}[section]

\begin{document}

\title{Euclidean and Hyperbolic Asymmetric Topological Quantum Codes}

\author{Clarice Dias de Albuquerque
\thanks{Science and Technology Center, Federal University of Cariri (UFCA), 63048-080, Juazeiro do Norte, CE, Brazil,
E-mail:~{\tt \small clarice.albuquerque@ufca.edu.br}.
}
Giuliano Gadioli La Guardia
\thanks{Department of Mathematics and Statistics,
State University of Ponta Grossa (UEPG), 84030-900, Ponta Grossa, PR, Brazil,
E-mail:~{\tt \small gguardia@uepg.br.}
}
Reginaldo Palazzo Jr. \\
\thanks{Department of Communications, FEEC-UNICAMP, 13083-852, Campinas, SP, Brazil,
E-mail:~{\tt \small palazzo@dt.fee.unicamp.br.}
}
Cátia Regina de Oliveira Quilles Queiroz
\thanks{Department of Mathematics, Institute of Exact Sciences, Federal University of Alfenas (UNIFAL-MG), 37130-000, Alfenas, MG, Brazil,
E-mail:~{\tt \small catia.quilles@unifal-mg.edu.br.}
}
Vandenberg Lopes Vieira
\thanks{Department of Mathematics, State University of Paraíba (UEPB), 58429-500, Campina Grande, PB, Brazil,
E-mail:~{\tt \small vandenberglv@uepb.edu.br.}
}}

\maketitle

\begin{abstract}
\noindent In the last three decades, several constructions of quantum error-correcting codes were presented in the literature. Among these codes, there are the asymmetric ones, i.e., quantum codes whose $Z$-distance $d_z$ is different from its $X$-distance $d_x$. The topological quantum codes form an important class of quantum codes, where the toric code, introduced by Kitaev, was the first family of this type. After Kitaev's toric code, several authors focused attention on investigating its structure and the constructions of new families of topological quantum codes over Euclidean and hyperbolic surfaces. As a consequence of establishing the existence and the construction of asymmetric topological quantum codes in Theorem \ref{main}, the main result of this paper, we introduce the class of hyperbolic asymmetric codes. Hence, families of Euclidean and hyperbolic asymmetric topological quantum codes are presented. An analysis regarding the asymptotic behavior of their distances $d_x$ and $d_z$ and encoding rates $k/n$ versus the compact orientable surface's genus is provided due to the significant difference between the asymmetric distances $d_x$ and $d_z$ when compared with the corresponding parameters of topological codes generated by other tessellations. This inherent unequal error-protection is associated with the nontrivial homological cycle of the $\{p,q\}$ tessellation and its dual, which may be appropriately explored depending on the application, where $p\neq q$ and $(p-2)(q-2)\ge 4$. Three families of codes derived from the $\{7,3\}$, $\{5,4\}$, and $\{10,5\}$ tessellations are highlighted.
\end{abstract}

\section{Introduction}

Several constructions and investigations concerning the structure of asymmetric quantum error-correcting codes (AQECCs) have been presented in the literature \cite{Evans:2007,Ezerman:2010,Ioffe:2007,LaGuardia:2011,LaGuardia:2012,LaGuardia:2012I,LaGuardia:2013,Sarvepalli:2008,Sarvepalli:2009,Stephens:2008,Wang:2010}.

Asymmetric quantum error-correcting codes are quantum codes defined over quantum channels where qudit-flip errors and phase-shift errors may have different probabilities. Steane introduced the notion of asymmetric quantum errors in \cite{Steane:1996}. As usual, ${[[n, k, d_{z}/d_{x}]]}_{q}$ denotes the parameters of an asymmetric quantum code, where $n$ is the code length, $q^k$ is the code dimension, $d_{z}$ is the minimum distance corresponding to phase-shift errors, and $d_{x}$ is the minimum distance corresponding to qudit-flip errors. The combined amplitude damping and dephasing channel (specific to binary systems; see \cite{Sarvepalli:2008}) is an example of a quantum channel that satisfies $d_{z} > d_{x}$, i. e., the probability of occurrence of phase-shift errors is greater than the probability of occurrence of qudit-flip errors.

In \cite{Evans:2007}, the authors explored the asymmetry between qubit-flip and phase-shift errors to perform optimization compared to quantum error-correcting codes (QECCs). In \cite{Ioffe:2007}, BCH and LDPC codes were employed to correct qubit-flip errors and phase-shift errors, respectively, of the corresponding AQECC. In \cite{Stephens:2008}, an investigation of AQECCs via code conversion was proposed. Asymmetric stabilizer codes derived from LDPC codes were constructed in \cite{Sarvepalli:2008}. In \cite{Sarvepalli:2009}, the authors derived quantum Singleton and linear programming bounds for families of AQECCs. In \cite{Wang:2010}, the construction of nonadditive AQECCs and constructions of asymptotically good AQECCs derived from algebraic-geometry codes were presented. In \cite{Ezerman:2010}, the Calderbank-Shor-Steane (CSS) construction \cite{Nielsen:2000,Calderbank:1998,Ketkar:2006} was extended to include codes endowed with the Hermitian and also trace Hermitian inner products. In \cite{LaGuardia:2011, LaGuardia:2012, LaGuardia:2012I}, one of us constructed families of AQECCs derived from BCH codes by expanding classical generalized Reed-Solomon codes and by applying product codes, respectively. In \cite{LaGuardia:2013}, the same author extended to asymmetric quantum codes the techniques of puncturing, extending, expanding, direct sum, and the $(u|u + v)$ construction.

On the other hand, topological quantum codes or surface codes are built on tessellations of two-dimensional manifolds, associating a qubit to each edge of the tessellation and the stabilizer operators to the vertex and faces, and stand out for being adaptable for fault-tolerant implementation. These codes were proposed first by Kitaev \cite{Kitaev:1997, Kitaev:2003}, and its construction in the square lattice of the torus earned it the name {\it toric code}. In \cite{Dennis:2001}, a more detailed analysis of Kitaev codes is made, observing the procedures for encoding, measurement, and performing fault-tolerant universal quantum computation with surface codes. A broader characterization of surface codes was done in \cite{Bombin:2007}, which are renamed as homological codes. In \cite{ClaPaBra:2009}, a generalization of toric codes was presented for surfaces $\mathcal{S}_{g}$ whose genus is $g \ge 2$. In this construction, the properties intrinsic to these surfaces are highlighted, and several code examples are tabulated.

Much other research stems from the initial proposal of topological quantum codes, such as the generalization to qudits, the search for calculating thresholds for various error models, decoding algorithms, and even another proposed construction of topological codes important enough, the color codes, \cite{Bombin:2006, ClaPaBra:2010, ClaPaBra:2014, Delfosse:2013, Terhal:2016}.

The way the code operators are defined allows the errors $X$ and $Z$ to be corrected independently, using primal and dual tessellation. Another essential feature of stabilizer generators is that, unlike other stabilizer codes, their supports are local; each operator involves few qubits in the code block. Since these qubits are close together, this implies more straightforward quantum gates, facilitating its potential physical implementation.

The minimum distance of the surface codes is defined as the minimum between the number of edges in the shortest homologically nontrivial cycle of primal tessellation and dual tessellation. It is clear from this definition that the minimum distance of these codes increases with the tessellation of the surface, with the number $n$ of qubits comprising the code. These properties of topological quantum codes make them strong candidates for controlling errors and also allowing utilizing them as a quantum memory, \cite{Kitaev:2003, Dennis:2001, Terhal:2016}.

Although Hansen~\cite{Hansen:2019} introduced and constructed asymmetric quantum codes on toric surfaces, i.e., a family of AQECCs with parameters $[[ (q-1)^{2}, k_1 + k_2 - (q-1)^{2}, d_z / d_x ]]_q$, where $q$ is a prime power, $0\leq b_i \geq a_i \geq q-2$, $i=1, 2$, $k_i =\frac{(q-1)(a_i + b_i +1 + \gcd (a_i - b_i, q-2)+1}{2}$, with $i=1, 2$, $d_z \geq q -1 - a_1$ and $d_x \geq q -1 - a_2$, his approach is quite different from the approach being proposed in this paper, since the former is based on the CSS construction whereas the latter is based on homology-cohomology.

This paper aims to show the existence and the construction of hyperbolic asymmetric topological quantum codes (ATQCs) and some families of toric ATQCs. In addition, to propose the construction of such families of codes based on typical concepts of homology and cohomology employed in topological quantum codes, as for example:
\begin{itemize}
\item $\left[\left[3\xi, \; 2, \; \xi \; / \; \lceil \xi \sqrt{3} \rceil \right]\right]$;
\item $\left[\left[\lambda^2, \; 2, \; \left \lceil \lambda / \sqrt{3} \right \rceil \; / \; \lambda \right]\right]$;
\item $\left[\left[\frac{6p(g-1)}{(p-6)}, \; 2g, \; d_z \; / \; d_x \right]\right]$;
\item $\left[\left[\frac{4p(g-1)}{p-4}, \; 2g, \; d_z \; / \; d_x \right]\right]$;
\item $\left[\left[\frac{10p(g-1)}{3p-10}, \; 2g, \; d_z \; / \; d_x \right]\right]$;
\item $\left[\left[\frac{3p(g-1)}{p-3}, \; 2g, \; d_z \; / \; d_x \right]\right]$,
\end{itemize}
where $\xi$ is a positive integer satisfying some conditions, $\{p, q\}$ denotes a regular hyperbolic tessellation, $g$ is the genus of the surface, and $d_z$, $d_x$, $d_h$ are given in Eqs.~(\ref{dxdy}) and (\ref{eqdh}), respectively. Recall that a $\{p,q\}$ tessellation consists of $q$ polygons each with $p$ edges meeting at each vertex. In this paper we consider tessellations whose fundamental polygons are regular. The first two families listed above are examples of toric ATQCs, while the remaining families are examples of hyperbolic ATQCs. We highlight the third family in the list related with the hyperbolic ATQCs derived from the $\{7,3\}$ tessellation whose parameters are $\left[\left[42(g-1), \; 2g, \; d_z \geq \left \lceil d_h / 1.0906 \right \rceil \; / \; d_x \geq \left \lceil d_h / 0.5663 \right \rceil \right]\right]$, due to the great difference between the distances $d_z$ and $d_x$ when comparing with the corresponding parameters of distinct topological codes. The fourth family, $\{p,4\}$, has the $\{5,4\}$ tessellation as its illustrious tessellation preserving most of the characteristics of the self-dual tessellation as employed in Kitaev's toric codes. However, the resulting hyperbolic code being capable of providing unequal-error protection when $g=4$ since for this value of $g$ $d_x=5$ and $d_z = 4$ therefore, leading to an ATQC code. The fifth family, $\{p,5\}$, has the classes of codes from the tessellation $\{10,5\}$ as a member of the $\{4i+2,2i+1\}$ tessellation whose encoding rate asymptotically goes to 1, \cite{ClaPaBra:2014}.

This paper is organized as follows. In Section~\ref{Sec2}, basic concepts of asymmetric codes and the group of errors are presented. In Section~\ref{Sec3}, the topological quantum codes are defined, and it is shown the main aspects of their construction on Euclidean and hyperbolic surfaces. Section~\ref{Sec4} is dedicated to basic concepts in homology and cohomology, fitting the purpose of this paper, an essential feature of topological quantum codes. In~Section~\ref{topcodes}, this paper's main contribution is established in Theorem \ref{main}, where the existence and construction of asymmetric topological quantum codes are shown. The construction of families of such codes is presented in Section~\ref{Sec6}. Finally, Section~\ref{Sec7} is reserved for the final remarks of the paper.

\section{Error Group and Asymmetric Codes} \label{Sec2}

This section reviews the asymmetric channel's error model and provides the necessary concepts related to this paper's content. For more detailed information on asymmetric quantum codes, we refer the reader to \cite{Sarvepalli:2009}. Although qubits are considered throughout this paper, the error operators are defined in a more general setting, i.e., qudits. We call attention to the fact that only in this section we employ the symbol $p$ as a prime number and $q$ is a prime power. Otherwise, $p$ and $q$ are integer numbers related to the tessellation $\{p,q\}$.

As it is usual in quantum theory, the scenario of our construction is the complex Hilbert space ${\cal H} = {\mathbb C}^{q^n} = {\mathbb{C}}^{q} \otimes \ldots \otimes {\mathbb{C}}^{q}$. The vectors $| x \rangle$ form an orthonormal basis of ${\mathbb{C}}^{q}$, where the labels $x$ are over the finite field ${\mathbb F}_{q}$. Recall that the trace map ${\operatorname{tr}}_{q^{m}/q}: {\mathbb F}_{q^{m}} \longrightarrow {\mathbb F}_{q}$ is defined as
${\operatorname{tr}}_{q^{m}/q}(a):= \displaystyle \sum_{i=0}^{m-1} a^{q^{i}}$.

If $a, b \in {\mathbb F}_{q}$, we define the unitary operators $X_a$ and $Z_b$ in ${\mathbb{C}}^{q}$, respectively, by
\[
X_a | x \rangle =\mid x + a\rangle \quad \mbox{and} \quad Z_b| x \rangle = w^{tr_{q/p}(bx)}| x\rangle,
\]
where $w=\exp (2\pi i/ p)$ is a $p$th root of unity.

Let ${\bf a}= (a_1, \ldots , a_n)$ and ${\bf b}= (b_1, \ldots , b_n)$ be vectors in ${\mathbb F}_{q}^{n}$. The tensor product of the error operators is then defined by
\[
X_{{\bf a}}= X_{{a_1}}\otimes \ldots \otimes X_{{a_n}} \quad \mbox{and} \quad Z_{{\bf b}}= Z_{{b_1}}\otimes \ldots \otimes Z_{{b_n}}.
\]

The set
\[
{\bf E}_{n} = \{ X_{{\bf a}}Z_{{\bf b}} \mid {\bf a}, {\bf b} \in {\mathbb F}_{q}^{n} \},
\]
is an \emph{error basis} in ${\mathbb C}^{q^n}$ and

\[
{\bf G}_{n} = \{w^c X_{{\bf a}}Z_{{\bf b}} \mid {\bf a}, {\bf b} \in {\mathbb F}_{q}^{n} , c\in {\mathbb F}_p \},
\]
is its associated \emph{error group}.

Given a quantum error ${\bf e} = w^c X_{{\bf a}}Z_{{\bf b}} \in {\bf G}_{n}$, the $X$-weight of ${\bf e}$ is defined as $\operatorname{wt}_{X}(e) = | \{i: 1\leq i\leq n \; |\; a_i \neq 0\} |$ and similarly, the $Z$-weight is defined as $\operatorname{wt}_{Z}(e) = | \{i: 1\leq i\leq n \; |\; b_i \neq 0\} |$.

An AQECC with parameters ${((n,K,d_{z}/ d_{x}))}_{q}$ is a $K$-dimensional subspace of ${\mathbb C}^{q^n}$ which corrects all qudit-flip errors up to $\lfloor \frac{d_{x}-1}{2} \rfloor$ and all phase-shift errors up to $\lfloor \frac{d_{z}-1}{2} \rfloor$. An ${((n, q^{k}, d_{z}/d_{x}))}_{q}$ AQECC is denoted by ${[[n, k, d_{z} / d_{x}]]}_{q}$.

\section{Topological Quantum Codes} \label{Sec3}

Topological quantum codes, as proposed by Kitaev, consider a square $l\times l$ as the fundamental polygon, $P^{\prime}$ which will be tiled by the fundamental polygon $P$ of the $\{4,4\}$ tessellation. This concept may be extended to any fundamental polygon, $P^{\prime}$, which will be tiled by proper $\{p,q\}$ tessellations, whose fundamental polygon is denoted by $P$, satisfying the condition that the ratio between the hyperbolic areas of $P^{\prime}$ and $P$, that is, $\mu (P^{\prime})/\mu (P)$, is an integer.

\begin{definition}
Let $P^{\prime}$ be a fundamental region of an orientable compact surface with genus $g$. Let $\{p,q\}$ be a tessellation which tiles $P^{\prime}$ with $E$ edges, $V$ vertice and $F$ faces. Given a vertex $v \in V$ and a face $f \in F$, the operators $X_f$ are defined as the tensor product of the operator $X$ corresponding to each edge forming the border of the face $f$ and the operators $Z_v$ as the tensor product of the operator $Z$ corresponding to each edge having $v$ as the common vertex. A topological quantum code $\cal{C}$ of length $n = |E|$, and stabilizer ${\cal{S}} = \{X_f | \ f \in F \}\cup \{Z_v | \ v \in V\}$, encodes $k = 2g$ qubits (if the surface has no border) and its minimum distance is $d = \operatorname{min}\{d_x, d_z\}$, where $d_x$ denotes the code distance in the $\{p,q\}$ tessellation, whereas $d_z$ denotes the code distance in the dual tessellation $\{q,p\}$.
\end{definition}

The minimum gives the distance of the topological quantum code between the number of edges in the shortest homologically nontrivial cycle of the primal tessellation and the dual tessellation, \cite{Dennis:2001, Lidar:2013}.

These operators constitute a Hamiltonian with local interactions $H_0 = - \sum_{f}X_f - \sum_{v}Z_v$, whose ground state coincides with the subspace protected by the code, \cite{Kitaev:2003}. The operators $X_f$ and $Z_v$ are Hermitian and have eigenvalues 1 and -1, so the difference between their eigenvalues is two, and the excited states are separated by an energy gap greater than or equal to 2, which allows detecting and correct errors on the physical level.

The first topological quantum code was proposed by Kitaev, \cite{Kitaev:1997}, the toric code $[[2l^2, 2, l]]$, whose qubits are in a one-to-one correspondence with the edges of the square lattice $l \times l$ ($\mathbb{Z}^2$ Euclidean lattice with $g=1$) or as it is known as \textit{flat torus}, whose fundamental polygon is tiled by the $\{4,4\}$ tessellation. In this tessellation $n=|E|=2l^2$ qubits and $k=2$. For each face $f$ and vertex $v$ of the lattice, the stabilizer operators are defined as
\begin{eqnarray} \label{operators}
X_f = \bigotimes_{j \in E_f} X^j \qquad \qquad Z_v = \bigotimes_{j \in E_v} Z^j,
\end{eqnarray}
where $E_f$ denotes the set of four edges forming the border of $f$ and $E_v$ denotes the set of four edges having $v$ as the common vertex, see Figure~\ref{fig1}. The toric code consists of the space which is fixed by these operators, ${\cal C}=\{|\psi \rangle \quad | \quad X_f |\psi
\rangle = |\psi \rangle \quad \mbox{and} \quad Z_v |\psi \rangle = |\psi \rangle\}$.

Since a homologically nontrivial cycle is an edge path in the tiling that can not be shrunk to a point, that is, a curve that is not the boundary of any region, it follows that on the flat torus, the shortest homologically nontrivial cycle corresponds exactly to the orthogonal axis of the tiling, see Figure~\ref{fig2}. It follows from this that $d = l$.
\begin{center}
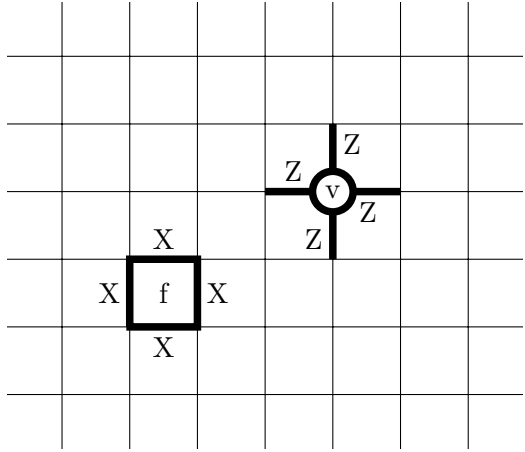
\begin{figure}[h!]
\centering
\begin{tikzpicture}[scale=.9]
\clip (0.2,0.2) -- (7.8,0.2) -- (7.8,6.8) -- (0.2,6.8) -- (0.2,0.2);
\draw (0,0) grid (8,8);
\draw[line width=.1cm] (1.95,2) -- (3,2) -- (3,3) -- (2,3) -- (2,2);
\node at (2.5,2.5){f};
\node at (1.7,2.5){X};
\node at (3.3,2.5){X};
\node at (2.5,1.7){X};
\node at (2.5,3.3){X};
\draw[line width=.1cm] (5,3) -- (5,4);
\draw[line width=.1cm] (4,4) -- (5,4);
\draw[line width=.1cm] (5,4) -- (5,5);
\draw[line width=.1cm] (5,4) -- (6,4);
\fill[white] (5,4) circle (0.3);
\draw[line width=.1cm] (5,4) circle (0.3);
\node at (5,4){v};
\node[right] at (5,4.7){Z};
\node[left] at (5,3.3){Z};
\node[left] at (4.7,4.3){Z};
\node[left] at (5.8,3.7){Z};
\end{tikzpicture}
\caption{$\{4,4\}$ tessellation on the torus and the face and vertex stabilizer operators $X$ and $Z$, respectively, \cite{Kitaev:1997}.}
\label{fig1}
\end{figure}
\end{center}

\begin{center}
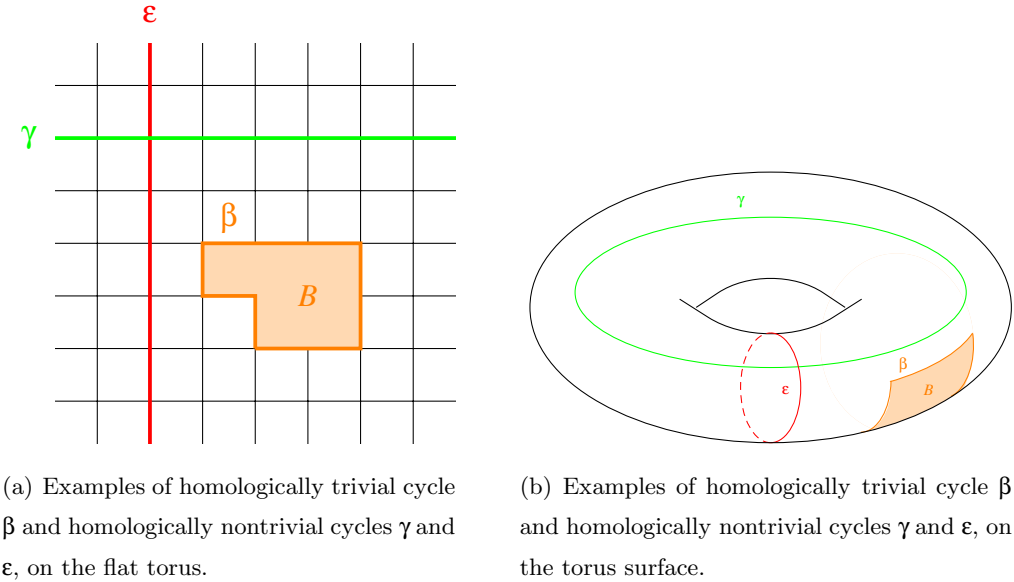
\begin{figure}[h!]
\centering
\subfigure[Examples of homologically trivial cycle $\beta$ and homologically nontrivial cycles $\gamma$ and $\varepsilon$, on the flat torus.]{
\begin{tikzpicture}[scale=.7]
\node[red, above] at (2,8) {\large{$\varepsilon$}};
\node[green, left] at (0,6) {\large{$\gamma$}};
\clip (0.2,0.2) -- (7.8,0.2) -- (7.8,7.8) -- (0.2,7.8) -- (0.2,0.2);
\draw (0,0) grid (8,8);
\draw[red, line width=0.05cm,] (2,0) -- (2,8);
\draw[green, line width=0.05cm] (0,6) -- (8,6);
\fill[orange!30] (4,2) -- (6,2) -- (6, 4) -- (3, 4) -- (3, 3) -- (4, 3) -- (4, 2);
\draw[orange, line width=0.05cm] (4,2) -- (6,2) -- (6, 4) -- (3, 4) -- (3, 3) -- (4, 3) -- (4, 2);
\node[orange] at (5,3) {\large{$B$}};
\node[orange, above] at (3.5,4) {\large{$\beta$}};
\end{tikzpicture}}
\hspace{0.6cm}
\subfigure[Examples of homologically trivial cycle $\beta$ and ho\-mologically nontrivial cycles $\gamma$ and $\varepsilon$, on the torus surface.]{
\begin{tikzpicture}[xscale=2, yscale=2]
\begin{scope}
\clip (0.84,-0.24) ellipse (0.51 and 0.6);
\fill[orange!30] (0,0) ellipse (1.6 and .9);
\fill[white] (0,0) ellipse (1.4 and .6);
\fill[white] (0.6,-0.53) ellipse (0.2 and 0.3);
\end{scope}
\draw[green] (0,0.1) ellipse (1.3 and .5);
\node[green] at (-0.2,.7) {\tiny{$\gamma$}};
\begin{scope}[scale=0.55]
\path[rounded corners=24pt] (-.9,0)--(0,.6)--(.9,0) (-.9,0)--(0,-.56)--(.9,0);
\draw[rounded corners=28pt] (-1.1,.1)--(0,-.6)--(1.1,.1);
\draw[rounded corners=24pt] (-.9,0)--(0,.6)--(.9,0);
\end{scope}
\draw[red, densely dashed] (0,-.9) arc (270:90:.2 and .365);
\draw[red] (0,-.9) arc (-90:90:.2 and .365);
\node[red] at (0.1,-.55) {\tiny{$\varepsilon$}};
\draw[orange] (0.6,-.83) arc (-90:-1:.2 and .348);
\draw[orange] (1.2,-.59) arc (-75:13:.2 and .348);
\node[orange] at (1.05,-.55) {\tiny{$B$}};
\node[orange] at (.88,-.38) {\tiny{$\beta$}};
\begin{scope}
\clip (1.12,-.36) ellipse (0.4 and .23);
\draw[orange] (0,0) ellipse (1.4 and .6);
\end{scope}
\draw (0,0) ellipse (1.6 and .9);
\end{tikzpicture}}
\caption{Distance of Kitaev's toric code.}
\label{fig2}
\end{figure}
\end{center}

The generalization of toric codes ($g=1$) to $g$-toric codes ($g\ge 2$) in a two-dimensional manifold, \cite{ClaPaBra:2009}, allows, for each fixed $g$, the construction of topological quantum codes with distinct characteristics according to the tiling of the fundamental polygon by the proper $\{p,q\}$ tessellation of this manifold, however maintaining the same useful properties of the toric codes such as the location of the stabilizer operators and the excellent scaling of the code distance. This generalization is justified because in \cite{Cavalcante:2005} and \cite{Brandani:2006}, it was shown that the performance (symbol/bit error probability) of a digital communication system in a 2D manifold decreases when the genus of the surface increases. Recent papers indicate that this expected efficiency is achieved on surfaces with $g \ge 2$. In \cite{Terhal:2016}, the authors suggested that codes derived from the $\{5,4\}$ tessellation may be better candidates for 2D quantum memory.

Let $P^{\prime}$ be a regular hyperbolic polygon whose $4g$ edges are pairwise identified and satisfy the edge and angle conditions \cite{ClaPaBra:2009}. These later conditions lead to a surface $\mathcal{S}_{g}$, which suits as planar models for $\mathcal{S}_{g}$. The simplest and most analyzed tessellations are the self-dual $\{4g,4g\}$ tessellation. For each genus $g$, $g\ge 2$, the hyperbolic topological quantum codes are built on by proper tilings of the $4g$ edges polygon viewed as the planar surface of genus $g$.

The $\{p,q\}$ hyperbolic tessellations which tile $P^{\prime}$, the fundamental region of a compact orientable surface $\mathcal{S}_{g}$, are directly related to an important topological invariant known as the Euler characteristic $\chi(\mathcal{S}_{g})$, given by,
\begin{eqnarray}
\chi(\mathcal{S}_{g})=|V|-|E|+|F|=2-2g.
\end{eqnarray}

All possible hyperbolic $\{p,q\}$ tessellations tiling $P^{\prime}$ must satisfy the following necessary and sufficient conditions, \cite{ClaPaBra:2009}:

\begin{itemize}
\item[(i)] $(p-2)(q-2)>4$, since considering a fundamental polygon (face) of the hyperbolic $\{p,q\}$ tessellation, and dividing it into $p$ triangles each with one common vertex at the center of the face, with an angle $\frac{2\pi}{p}$ and the other two equal to $\frac{\pi}{q}$. As the sum of the internal angles of a hyperbolic triangle is less than $\pi$, \cite{Firby:1991}, the inequality follows;
\item[(ii)] the number of faces of the tessellation is a positive integer $n_f= \dfrac{\mu(P^{\prime})}{\mu(P)}$, where $\mu(P^{\prime})$ is the area of the polygon $P^{\prime}$ and $\mu(P)$ is the area of the fundamental polygon associated with the $\{p,q\}$ tessellation;
\item[(iii)] the number of faces $n_f^*$ of the dual tessellation is also a positive integer.
\end{itemize}

From the Gauss-Bonnet Theorem and from the first two previous conditions, the number of faces is
\begin{eqnarray} \label{eqarea}
n_f=\dfrac{4q(g-1)}{pq-2p-2q}.
\end{eqnarray}

The stabilizer operators are similarly defined as in Kitaev's construction: given a face $f$ of the tiling, the face operator $X_f$ acts non-trivially on the $p$ qubits corresponding to the edges of the giving face. Given a vertex $v$ of the tiling, the vertex operator $Z_v$ acts non-trivially on the $q$ qubits whose edges have $v$ as the common vertex. Thus,

$$X_f = \bigotimes_{j \in E_f} X^j \qquad \qquad Z_v = \bigotimes_{j \in E_v} Z^j,$$
where $E_f$ denotes the set of edges forming the face $f$, and $E_v$ the set of edges having $v$ as a common vertex. Hence, code ${\cal C}$ is the space stabilized by these operators.

Note that $\displaystyle\prod_{f \in F} X_f = 1$, since each edge belongs to exactly two faces and $X^2=I$. Similarly, $\displaystyle\prod_{v \in V} Z_v = 1$. Thus, a face operator and a vertex operator can be written as a combination of the same type's remaining operators. Therefore, the number of stabilizer generators is the number of face operators plus the number of vertex operators minus 2, i.e., $|V| + |F| -2$.

The parameters of the hyperbolic topological quantum code $[[n,k,d]]$ are:
\begin{itemize}
\item $n = |E| = n_f \frac{p}{2}$, since each edge of the $\{p,q\}$ tessellation belongs to two faces simultaneously;
\item $k = 2g$, because $n-k$ is the number of stabilizer generators, that is, $| V | + | F | - 2$. Then, $k = |E| - | V | - | F | + 2 = -\chi(\mathcal{S}_{g}) + 2 = -2+2g + 2=2g$.
\item Consequently, the code stabilized by these operators has dimension $2^{2g}$.
\item The code distance is $d = min\{d_x, d_z\}$, with
\begin{equation}\label{dxdy}
d_x = \left \lceil \frac{d_h}{l (p,q)} \right \rceil \qquad \mbox{and} \qquad d_z = \left \lceil \frac{d_h}{l (q,p)} \right \rceil,
\end{equation}
where
\begin{itemize}
\item  $d_h$ is the hyperbolic distance between opposite edge-pairings of $P^{\prime}$:
\begin{eqnarray} \label{eqdh}
d_h = 2 \ \textrm{arccosh} \left(\frac{\cos (\pi/4g)}{\sin (\pi/4g)}\right);
\end{eqnarray}

\item $l(p,q)$ is the edge length of $\{p,q\}$:
\begin{eqnarray} \label{eql}
l(p,q) = \textrm{arccosh} \left(\frac{\cos^2 (\pi/q) + \cos (2 \pi/p)}{\sin^2 (\pi/q)}\right).
\end{eqnarray}

\end{itemize}
\end{itemize}

\section{Homology and Cohomology} \label{Sec4}

In this section, some basic concepts from homology, which will be considered in Section~\ref{topcodes}, are reviewed. For more detailed information we refer the reader to \cite{Elon:2012,Vick:1994}. Although in the whole paper we use $\bar{c}$ to denote the homological (cohomological) equivalence class of a chain (cochain) $c$, in this section we write $[c]$ in order to keep the standard notation of homology theory.

Let $R$ be a commutative ring with identity. A \emph{complex of chains} with coefficients over $R$ is a sequence ${\mathcal C}= (C_l , {\partial}_{l})$ of $R$-modules $C_l$, where $l \geq 0$ is an integer and ${\partial}_{l}: C_l \longrightarrow C_{(l-1)}$ are $R$-homomorphisms, called \emph{boundary operators}, satisfying ${\partial}_{l}\circ {\partial}_{(l+1)}=0$. Each element $c \in C_l$ is called an $l$-\emph{chain}. If ${\partial}_{l}c=0$, we say that $c$ is an $l$-\emph{cycle} or simply a \emph{cycle} when there is no possibility of confusion. The set $Z_l$ of all $l$-cycles is a submodule of $C_l$: in fact, $Z_l$ is the kernel of ${\partial}_{l}: C_l \longrightarrow C_{(l-1)}$.

If $ b = {\partial}_{(l+1)} c$, we say that $b$ is the \emph{boundary} of the $(l+1)$-chain $c$. The set $B_l$ of $l$-chains that are boundaries is a submodule of $C_l$: $B_l$ is the image of the homomorphism ${\partial}_{(l+1)}: C_{(l+1)} \longrightarrow C_{l}$. The fundamental equality ${\partial}_{l}\circ {\partial}_{(l+1)}=0$ means that every boundary is a cycle, i.e., $B_l \subset Z_l$. In this context, the $R$-module quotient
\[
H_l({\mathcal C})= H_l := Z_l /B_l
\]
is called the $l$-\emph{homology group} of ${\mathcal C}$ with coefficients over $R$. The elements of $H_l$ are homology classes of cycles $z \in Z_l$:
\[
 [z] = z + B_l = \{ z + {\partial}_{(l+1)} c ;\quad c \in C_{l+1}\}.
\]

The equivalence relation is given as follows: if $z, \widetilde{z} \in Z_l$ then $[z] = [\widetilde{z}]$ if and only if $\widetilde{z}-z= {\partial}_{(l+1)}c$ for some $c \in C_{(l+1)}$. The cycles $z, \widetilde{z}$ are called \emph{homologically equivalent}.

Analogously to complex chains, we may define the complex of cochains.

A \emph{complex of cochains} with coefficients over $R$ is a sequence ${\mathcal C}= (C^l , {\delta}_{l})$ of $R$-modules $C^l$, where $l \geq 0$ is an integer and ${\delta}_{l}: C^l \longrightarrow C^{(l+1)}$ are $R$-homomorphisms, called \emph{coboundary operators}, satisfying ${\delta}_{(l+1)}\circ {\delta}_{l}=0$. The elements of $C^l$ are called $l$-\emph{cochains}. If ${\delta}_{l}c^{*}=0$, then $c^{*}$ is said to be an $l$-\emph{cocycle}.

The set $Z^l$ of $l$-cocycles is a submodule of $C^l$; $Z^l = \operatorname{Ker}({\delta}_{l})$. Moreover, $B^l = \operatorname{Im}({\delta}_{(l-1)})$; the equality ${\delta}_{(l+1)}\circ {\delta}_{l}=0$ means that $B^l \subset Z^l$. The quotient $H^l := Z^l /B^l$ is called the $l$-\emph{cohomology group} of ${\mathcal C}$ with coefficients over $R$.

The equivalence class of some element $u \in Z^l$ is the set $[c^{*}]=\{c^{*}+ {\delta}\widetilde{c} ; \widetilde{c} \in C^{(l-1)}\}$, where the equivalence relation is given by $[c_1^{*}]=[c_2^{*}]$ if and only if $c_1^{*}-c_2^{*}= {\delta}c^{*}$ for some $c^{*} \in C^{(l-1)}$. In this case, $c_1^{*}$ and $c_2^{*}$ are called \emph{cohomologically equivalent}.

Throughout this paper we always consider $R= {\mathbb Z}_2$. Given a $\{p, q\}$ tessellation with vertex set $V$, edges set $E$ and faces set $F$, embedded in $\mathcal{S}_{g}$, we can represent each set of edges $E^{'}\subset E$ as a formal sum (called $1$-chains, denoted by $C_1$, according to the notation defined previously) $c = \displaystyle\sum_{i=1}^{|E|} c_i e_i$, where $c_i = 0$ if $e_i \notin E^{'}$, and $c_i = 1$ if $e_i \in E^{'}$.

Given two $1$-chains $c = \displaystyle\sum_{i=1}^{|E|} c_i e_i$ and $\widetilde{c} = \displaystyle\sum_{i=1}^{|E|} {\widetilde{c}}_i e_i$, the sum $c + \widetilde{c}$ is also a $1$-chain defined by the componentwise ${\mathbb Z}_2$ sums, i.e.,
\[
c + \widetilde{c}= \displaystyle\sum_{i=1}^{|E|} (c_i + {\widetilde{c}}_i) e_i.
\]
By definition, it is immediate to see that $C_1$ is an Abelian group isomorphic to ${\mathbb Z}_2^{|E|}$. Similarly, the set of $0$-chains $C_0$ (formal sum of vertices) and the $2$-chains $C_2$ (formal sum of faces) are also Abelian groups isomorphic, respectively, to ${\mathbb Z}_2^{|V|}$ and ${\mathbb Z}_2^{|F|}$.

In this paper it is sufficient to consider only the boundary operators ${\partial}_{2}: C_2 \longrightarrow C_{1}$ and ${\partial}_{1}: C_1 \longrightarrow C_{0}$. The first homology group $H_1$ is isomorphic to ${\mathbb Z}_2^{2g}$, where $g$ is the genus of the compact surface on which the tessellation is embedded. Moreover, from the Homology Theory, one has $C_1\simeq C^1$ and $H_1 \simeq H^1$.

\section{Asymmetric Topological Quantum Codes}\label{topcodes}

In this section, we present the main contribution of this paper. More precisely, we show how to construct ATQCs when considering tessellations in the hyperbolic plane. Throughout this paper we adopt that bit-flip errors are corrected in the original tessellation, whereas the phase-shift errors are corrected in the dual one (see Remark \ref{remark3}). The latter occurs due to applying the Hadamard gate in each qubit of the code, generating a new code in the dual tessellation. Our main result is given in the sequence.

\begin{theorem}\label{main}
Let $\{p, q\}$ be a nonself-dual tessellation such that $(p-2)(q-2) \ge 4$. Then there exists an $[[n,\; 2g,\; d_{z}/ d_{x}]]$ ATQC, where $n$ is the number of edges of the tessellation, $g \geq 1$ is the genus of a compact orientable surface, $d_{x}\geq \lceil d_g / l(p, q)\rceil$ and $d_{z}\geq \lceil d_g / l(q, p)\rceil$, where $d_g$ is the Euclidean distance $d_e$ of opposite sides of fundamental region of $\mathcal{S}_g$ for $g=1$ and $d_g$ is the hyperbolic distance $d_h$ of opposite sides of fundamental region of $\mathcal{S}_g$ for $g \ge 2$.
\end{theorem}

\begin{proof}
Let $P^{\prime}$ be the fundamental region of a compact orientable surface $\mathcal{S}_g$ of genus $g$. Let $\{p,q\}$ be a (Euclidean or hyperbolic) tessellation which tiles $P^{\prime}$. Let $|V|$ denote the cardinality of the vertices set, $|E|$ the cardinality of the edges set, and $|F|$ the cardinality of the faces set of such tiling. As seen previously, the corresponding Euler characteristic is given by $\chi (\mathcal{S}_g) = |V| - |E| + |F| = 2 - 2g$.

Similarly to Kitaev's toric code \cite{Kitaev:2003}, a qubit is attached to each edge. Since $C_1$ has $2^{|E|}$ elements, it follows that the vectors of the computational basis can be identified with the $1$-chains $c \in C_1$, i.e., $|c\rangle := \displaystyle\otimes_{i=1}^{|E|}|c_i\rangle$, $c_i =0, 1$. The proof follows the same line of reasoning as in \cite{Lidar:2013} Chapter 19, however with the inclusion of the necessary concepts to the new context. We begin by writing $X_c := \displaystyle\otimes_{i} X_i^{c_i}$ and $Z_c := \displaystyle\otimes_{i} Z_i^{c_i}$. We define a (Euclidean or hyperbolic) ATQC $Q_{{\mathcal H}}$ with basis elements
$$ | \overline{z}\rangle := \displaystyle\sum_{b \in B_1} |z + b\rangle,$$
where $\overline{z}$ is an element of $H_1$. Since $|H_1|=2^{2g}$, it follows that $Q_{{\mathcal H}}$ has dimension $2^{2g}$. Let $S= \langle \{X_f, Z_v ; \ \forall \ v, f \}\rangle$. Note that these operators commute among themselves for all $v \in V$ and for all $ f \in F$: the $X_f$'s operators commute among themselves and the $Z_v$'s operators commute among themselves; for every $X_f$ and $Z_v$, there exist either zero or two edges in common and the result follows since $X$ and $Z$ anticommute.

We will show that $Q_{{\mathcal H}}$ is the code stabilized by $S$. If $c_1 , c_2 \in C_1$ then $X_{c_2} | c_1 \rangle = | c_1 + c_2 \rangle$, because the action of $X_{c_2}$ is to change the computational basis vectors. Thus, if $c, z \in Z_1$ then

\[
\begin{array}{ll}
X_c |\overline{z}\rangle & = X_c \left(\displaystyle\sum_{b \in B_1} |z + b\rangle \right) = \displaystyle\sum_{b \in B_1} X_c |z + b\rangle \\
  &  \\
 & = \displaystyle\sum_{b \in B_1} | z + b + c\rangle = |\overline{z} + \overline{c}\rangle.
\end{array}
\]

Let $c_f$ be the $1$-chain corresponding to the edges belonging to the face $f$. Thus,
\begin{eqnarray}
X_f |\overline{z}\rangle = \displaystyle\sum_{b \in B_1} |z + c_{f} + b\rangle
= \displaystyle\sum_{\widetilde{b} \in B_1} |z + \widetilde{b}\rangle = |\overline{z}\rangle.
\end{eqnarray}
Consequently, $X_f |\overline{z}\rangle= |\overline{z}\rangle$ for all $f \in F$. Analogously,
\begin{eqnarray}\label{eq3}
Z_v |\overline{z}\rangle = \displaystyle\sum_{b \in B_1} Z_v |z + b\rangle =
\displaystyle\sum_{b \in B_1} |z + b\rangle = |\overline{z}\rangle,
\end{eqnarray}
where the second equality in~(\ref{eq3}) follows from the fact that $Z_v$ acts nontrivially on either zero or two edges of the state $|z + b\rangle$.

As seen in Section \ref{Sec3}, the code stabilized by $S$ has dimension $2^{2g}$, hence $Q_{{\mathcal H}}$ is this code.

The next step is to show that the minimum $X$-distance of the code satisfies $d_{x}\geq \lceil d_g / l(p, q)\rceil$. For this, we need to find the logical operators for the bit-flip errors. Assume that $c \in C_1$ and $z \in Z_1$. If ${\partial}_1 c \neq 0$, and since ${\partial}_1$ is a group homomorphism, one has
\begin{eqnarray}\label{eq.4}
{\partial}_1 (z + c) = {\partial}_1 z + {\partial}_1 c = {\partial}_1 c \neq 0,
\end{eqnarray}
which implies that $z + c$ is not a cycle, i.e., $X_c | z\rangle $ maps $z \in Z_1$ outside of the code.

If $c \in B_1$, then we obtain
\begin{eqnarray}\label{eq.5}
X_c |\overline{z}\rangle = | \overline{z} + \overline{c}\rangle = | \overline{z}\rangle,
\end{eqnarray}
and there is no error to be corrected by the code.

If $ c \in Z_1 - B_1$, then $X_c | \overline{z}\rangle = | \overline{z} + \overline{c}\rangle$. In particular, taking $\overline{z}=\overline{c}$ it follows that $X_c |\overline{z}\rangle =\overline{0}$, i.e., $X_c$ sends a homologically nontrivial codeword to a homologically trivial codeword, and this error cannot be detected by the code. Another way of viewing that homologically nontrivial cycles $\overline{X}$ are the logical operators for bit-flip errors is to note that $\overline{X}$ commutes with all the generators $X_f$ and $Z_v$, for all $v, f$ of the stabilizer group but does not belong to the stabilizer, i.e., $\overline{X} \in N(S)-S$. Thus, $d_{x}\geq \lceil d_g / l(p, q)\rceil$.

We next analyze the dual tessellation to verify the correction of phase-shift errors. Recall that the Hadamard gate is represented by the matrix
\begin{eqnarray*}
H =\frac{1}{\sqrt{2}}\left[\begin{array}{c c}
1 & 1\\
1 & -1
\end{array}\right].
\end{eqnarray*}

Similarly to the Calderbank-Shor-Steane (CSS) construction, \cite{Calderbank:1998, Nielsen:2000}, to detect phase-shift errors, Hadamard gates $H$ are applied to each qubit. The generators of the stabilizer $S$ take the form
\[
H^{\otimes |E|}X_f H^{\otimes |E|} = \displaystyle\prod_{e \in {\partial}_2 f} Z_e:= Z_{f^{*}},
\]
and
\[
H^{\otimes |E|}Z_v H^{\otimes |E|} = \displaystyle\prod_{e| v \in {\partial}_1 e} X_e := X_{v^{*}}.
\]

Let $S^{*}=\langle \{X_{v^{*}}, Z_{f^{*}}; \ \forall \ v^{*}, f^{*}\}\rangle$. We will show that $H^{\otimes |E|}(Q_{{\mathcal H}})$ is the code stabilized by $S^{*}$.

Since $| \varphi\rangle \in Q_{{\mathcal H}}$, it follows that
\[
\begin{array}{ll}
Z_{f^{*}}(H^{\otimes |E|}(|\varphi\rangle))& = H^{\otimes |E|}X_f H^{\otimes |E|}(H^{\otimes |E|}(|\varphi\rangle))\\
 & \\
  & = H^{\otimes |E|}X_f(|\varphi\rangle)=H^{\otimes |E|}(|\varphi\rangle)
\end{array}
\]
and
\[
\begin{array}{ll}
X_{v^{*}}(H^{\otimes |E|}(|\varphi\rangle)) & = H^{\otimes |E|}Z_v H^{\otimes |E|}(H^{\otimes |E|}(|\varphi\rangle))\\
  &  \\
  & = H^{\otimes |E|}Z_v(|\varphi\rangle)=H^{\otimes |E|}(|\varphi\rangle).
\end{array}
\]

First, note that the number of $X$'s operators in the face $f$ is equal to the number of $Z$'s operators incident in the vertex $f^{*}$; analogously, the number of $Z$'s in the vertex $v$ equals the number of $X$'s in $v^{*}$. Additionally, from the homology theory, we know that $C_1$ is isomorphic to $C^{1}$, which implies $|C_1| = | C^{1}|$. This means that we can consider the code $H^{\otimes |E|}(Q_{{\mathcal H}})$ as defined in the dual tessellation $\{q, p\}$.

We must show that $[X_{v^{*}}, Z_{f^{*}}] = 0$ for all $v^{*}, f^{*}$. Similarly to the original tessellation, the operators $X_{v^{*}}$'s commute among themselves, the operators $Z_{f^{*}}$ commute among themselves and for all $\ v^{*}, f^{*}$, $X_{v^{*}}$ and $Z_{f^{*}}$ also commute, since they have either zero or two edges in common. By the same previous reasoning, we can see that $\displaystyle\prod_{v^{*} \in F^{*}} X_{v^{*}} = 1$ and $\displaystyle\prod_{f^{*} \in V^{*}} Z_{f^{*}} = 1$. Hence, the code stabilized by $S^{*}$ has dimension $2^{2g}$. Since $H^{\otimes |E|}$ is unitary, it follows that $H^{\otimes |E|}(Q_{{\mathcal H}})$ has also dimension $2^{2g}$. Therefore, $H^{\otimes |E|}(Q_{{\mathcal H}})$ is, in fact, the code stabilized by $S^{*}$.

Applying the Hadamard gate in a string of homologically nontrivial cycle $\overline{Z}$ (of phase errors) in the original tessellation becomes $H^{\otimes |E|}\overline{Z}H^{\otimes |E|}= \overline{X}$. Then it is easy to see that $\overline{X}$ commutes with all the generators of the stabilizer $S^{*}$, and it does not belong to $S^{*}$, i.e., $\overline{X}\in N(S^{*})-S^{*}$. Hence, $\overline{Z}$ in the original tessellation is a logical operator in the dual tessellation. Since the group, $H_1 = Z_1 / B_1$ is isomorphic to $H^{1}= Z^{1}/B^{1}$; it follows that the isomorphism sends (homologically) nontrivial cycles of the original tessellation into nontrivial cycles of the dual tessellation, which implies that the unique type of logical operators for phase-shift errors are such strings $\overline{Z}$ (in the original tessellation) of nontrivial cycles. Hence, $d_{z}\geq \lceil d_g / l(q, p)\rceil$, where $d_g$ is the Euclidean or hyperbolic distance of opposite sides of a fundamental region of $\mathcal{S}_g$ according to the genus of the surface, and the proof is complete.
\end{proof}

\section{Families of Asymmetric Topological Quantum Codes} \label{Sec6}

In this section, families of asymmetric topological quantum codes obtained from the construction described in Theorem~\ref{main} will be presented. In the first subsection, three families of toric ATQCs are derived by employing a parallelogram as the fundamental polygon $P^{\prime }$ being tiled by the $\{6,3\}$ tessellation. In the second subsection, families of hyperbolic ATQCs are derived.

\subsection{Families of toric ATQCs} \label{subsec6.1}

Let us consider a planar torus having either a square or a parallelogram as fundamental polygon $P^{\prime }$ being tiled by a tessellation $\{p,q\}$. A face of such a tessellation consists of a fundamental polygon $P$ with $p$ edges. Partitioning $P$ into $p$ triangles having a common vertex at the center of $P$, and that $q$ such faces meet at each vertex of $P$, it follows that each triangle has one of its angles equal to $\frac{2\pi}{p}$ and the other two equal to $\frac{\pi}{q}$. Since the sum of the internal angles of a triangle in Euclidean geometry equals $\pi$, it follows that $(p-2)(q-2)=4$. The solutions to this equation are the self-dual tessellation $\{4,4\}$, and also the $\{6,3\}$ tessellation and its dual $\{3,6\}$.

It is clear from Kitaev's toric code derived from the self-dual tessellation $\{4,4\}$ that the primal and dual tessellations' distances are equal. However, this is not the case when considering Kitaev's toric code derived from the dual tessellations $\{6,3\}$ and $\{3,6\}$ leading to the construction of two families of toric asymmetric topological quantum codes as established in Theorem~\ref{torichex1} and Theorem~\ref{torichex2}.

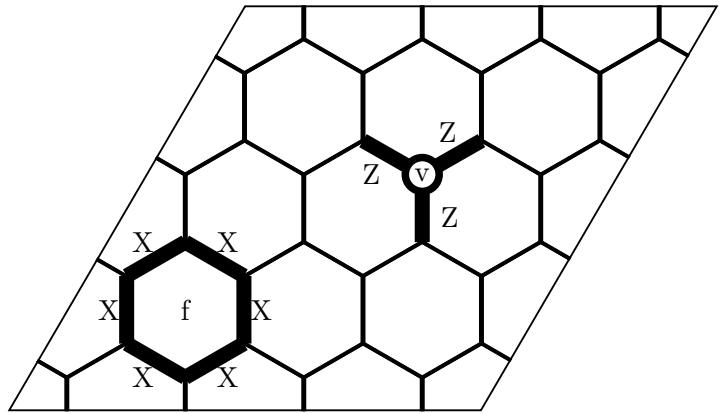
\begin{figure}[h!]
\centering
\begin{tikzpicture}[scale=.45]
\clip (-10.5,-12) -- (3.5,-12) -- (10.5,0) -- (-3.5,0) -- (-10.5,-12);
\draw[line width=.05cm] (-10.5,-12) -- (3.5,-12) -- (10.5,0) -- (-3.5,0) -- (-10.5,-12);
\foreach \x in {30,90,...,330}
{\draw[line width=.05cm] (\x:2) -- (60+\x:2);}
\begin{scope}[xshift=-3.5cm,yshift=0cm]
\foreach \x in {30,90,...,330}
{\draw[line width=.05cm] (\x:2) -- (60+\x:2);}
\end{scope}
\begin{scope}[xshift=3.5cm,yshift=0cm]
\foreach \x in {30,90,...,330}
{\draw[line width=.05cm] (\x:2) -- (60+\x:2);}
\end{scope}
\begin{scope}[xshift=7cm,yshift=0cm]
\foreach \x in {30,90,...,330}
{\draw[line width=.05cm] (\x:2) -- (60+\x:2);}
\end{scope}
\begin{scope}[xshift=10.5cm,yshift=0cm]
\foreach \x in {30,90,...,330}
{\draw[line width=.05cm] (\x:2) -- (60+\x:2);}
\end{scope}

\begin{scope}[xshift=-5.25cm,yshift=-3cm]
\foreach \x in {30,90,...,330}
{\draw[line width=.05cm] (\x:2) -- (60+\x:2);}
\end{scope}
\begin{scope}[xshift=-1.75cm,yshift=-3cm]
\foreach \x in {30,90,...,330}
{\draw[line width=.05cm] (\x:2) -- (60+\x:2);}
\end{scope}
\begin{scope}[xshift=1.75cm,yshift=-3cm]
\foreach \x in {30,90,...,330}
{\draw[line width=.05cm] (\x:2) -- (60+\x:2);}
\end{scope}
\begin{scope}[xshift=5.25cm,yshift=-3cm]
\foreach \x in {30,90,...,330}
{\draw[line width=.05cm] (\x:2) -- (60+\x:2);}
\end{scope}
\begin{scope}[xshift=8.75cm,yshift=-3cm]
\foreach \x in {30,90,...,330}
{\draw[line width=.05cm] (\x:2) -- (60+\x:2);}
\end{scope}

\begin{scope}[xshift=-3.5cm,yshift=-6cm]
\foreach \x in {30,90,...,330}
{\draw[line width=.05cm] (\x:2) -- (60+\x:2);}
\end{scope}
\begin{scope}[xshift=-7cm,yshift=-6cm]
\foreach \x in {30,90,...,330}
{\draw[line width=.05cm] (\x:2) -- (60+\x:2);}
\end{scope}
\begin{scope}[xshift=0cm,yshift=-6cm]
\foreach \x in {30,90,...,330}
{\draw[line width=.05cm] (\x:2) -- (60+\x:2);}
\end{scope}
\begin{scope}[xshift=3.5cm,yshift=-6cm]
\foreach \x in {30,90,...,330}
{\draw[line width=.05cm] (\x:2) -- (60+\x:2);}
\end{scope}
\begin{scope}[xshift=7cm,yshift=-6cm]
\foreach \x in {30,90,...,330}
{\draw[line width=.05cm] (\x:2) -- (60+\x:2);}
\end{scope}

\begin{scope}[xshift=-8.75cm,yshift=-9cm]
\foreach \x in {30,90,...,330}
{\draw[line width=.05cm] (\x:2) -- (60+\x:2);}
\end{scope}
\begin{scope}[xshift=-5.25cm,yshift=-9cm]
\foreach \x in {30,90,...,330}
{\draw[line width=.2cm] (\x:2) -- (60+\x:2);}
\end{scope}
\begin{scope}[xshift=-1.75cm,yshift=-9cm]
\foreach \x in {30,90,...,330}
{\draw[line width=.05cm] (\x:2) -- (60+\x:2);}
\end{scope}
\begin{scope}[xshift=1.75cm,yshift=-9cm]
\foreach \x in {30,90,...,330}
{\draw[line width=.05cm] (\x:2) -- (60+\x:2);}
\end{scope}
\begin{scope}[xshift=5.25cm,yshift=-9cm]
\foreach \x in {30,90,...,330}
{\draw[line width=.05cm] (\x:2) -- (60+\x:2);}
\end{scope}

\begin{scope}[xshift=-10.5cm,yshift=-12cm]
\foreach \x in {30,90,...,330}
{\draw[line width=.05cm] (\x:2) -- (60+\x:2);}
\end{scope}
\begin{scope}[xshift=-7cm,yshift=-12cm]
\foreach \x in {30,90,...,330}
{\draw[line width=.05cm] (\x:2) -- (60+\x:2);}
\end{scope}
\begin{scope}[xshift=-3.5cm,yshift=-12cm]
\foreach \x in {30,90,...,330}
{\draw[line width=.05cm] (\x:2) -- (60+\x:2);}
\end{scope}
\begin{scope}[xshift=0cm,yshift=-12cm]
\foreach \x in {30,90,...,330}
{\draw[line width=.05cm] (\x:2) -- (60+\x:2);}
\end{scope}
\begin{scope}[xshift=3.5cm,yshift=-12cm]
\foreach \x in {30,90,...,330}
{\draw[line width=.05cm] (\x:2) -- (60+\x:2);}
\end{scope}

\node at (-5.25,-9){f};
\node at (-3,-9){X};
\node at (-7.5,-9){X};
\node at (-6.5,-7){X};
\node at (-6.5,-11){X};
\node at (-4,-7){X};
\node at (-4,-11){X};
\draw[line width=.2cm] (1.75,-5) -- (1.75,-7);
\draw[line width=.2cm] (1.75,-5) -- (3.5,-4);
\draw[line width=.2cm] (1.75,-5) -- (0,-4);
\fill[white] (1.75,-5) circle (0.5);
\draw[line width=.1cm] (1.75,-5) circle (0.5);
\node at (1.75,-5){v};
\node[right] at (2,-6.25){Z};
\node at (2.5,-3.75){Z};
\node at (0.25,-5){Z};
\end{tikzpicture}
\caption{Fundamental polygon $P^{\prime }$ tiled by the $\{6,3\}$ tessellation. It is also illustrated the face and vertex stabilizer operators $X$ and $Z$, respectively.}
\label{fig4}
\end{figure}

In the next two subsections, three families of toric ATQCs are derived. We recall that the stabilizer operators are defined similarly as in Kitaev's codes, as shown in Figure~\ref{fig4}.

\subsubsection{Toric ATQCs with $(2\xi)a \times (2\xi)a$ parallelogram as the fundamental polygon}

\begin{corollary}\label{torichex1}
Let $P^{\prime }$ be the fundamental polygon of the torus planar model. Let $\{6,3\}$ and $\{3,6\}$ be the hexagonal tessellation and its dual tiling $P^{\prime }$, respectively. If $L=(2\xi)a$ is the edge length of $P^{\prime }$, where $a$ is the apothem of the regular hexagon and $\xi \in \mathbb{Z}_+^*$, then there exists a family of toric asymmetric topological quantum codes with parameters $[[3\xi, 2, \xi \; / \; \lceil \xi \sqrt{3} \rceil]]$.
\end{corollary}

\begin{proof}
Let $P^{\prime }$ be a parallelogram with edge length $L=2\xi a$, where $\xi \in \mathbb{Z}_+^*$; see Figure~\ref{fig5}. The height $h$ of $P^{\prime }$ is $h=(2\xi a)\dfrac{\sqrt{3}}{2}=\xi a \sqrt{3}$ and the angle $\alpha$ is equal to $\dfrac{\pi}{3}$. Since $a=\dfrac{l \sqrt{3}}{2}$, it follows that $h=\dfrac{\xi 3 l}{2}$.

Thus, the area  $A_{P^{\prime }}$ of $P^{\prime }$ is
\begin{eqnarray} \label{areaP}
A_{P^{\prime }}=\dfrac{\xi^2 3l^2 \sqrt{3}}{2}.
\end{eqnarray}

Let $P$ be the fundamental polygon of the $\{6,3\}$ tessellation. Let $A_{P}$ be the hexagon area. Then,
\begin{eqnarray} \label{areaH}
A_P=\dfrac{3l^2 \sqrt{3}}{2}.
\end{eqnarray}

Therefore, the number of faces $n_f$ of the tessellation is an integer number
\begin{eqnarray} \label{nfHa}
n_f=\dfrac{A_{P^{\prime }}}{A_P}=\xi^2.
\end{eqnarray}

The code length $n$ is the number of edges of the $\{6,3\}$ tessellation which tiles $P^{\prime }$. Since each edge belongs simultaneously to two faces, it follows that $n=n_f \dfrac{p}{2}$, e.g., the code length is $n=3 \xi^2$. Recall that the number of encoded qubits is $k=2g$. For the toric code $g=1$; so $k=2$.

Let $d_x$ be the code distance relative to the $X$-errors when $P^{\prime }$ is tiled by the $\{6,3\}$ tessellation. We call attention to the fact that the definition of the hyperbolic topological quantum codes still applies in the Euclidean case, when Euclidean tessellations are considered. Hence,
\begin{eqnarray} \label{dxHaformula}
d_x =\left \lceil \frac{d_e}{l(6,3)} \right \rceil,
\end{eqnarray}
where $d_e$ is the Euclidean distance between the opposite edge-pairings of the parallelogram. In this case, $d_e=L=2\xi a$. Since the hexagon edge length, denoted by $l(6,3)$, is $l$, it follows that
\begin{eqnarray} \label{dxHa}
d_x = \lceil {\xi \sqrt{3}} \rceil.
\end{eqnarray}

The area $A_{P^{\prime }}$ of $P^{\prime }$ remains the same when considering the $\{3,6\}$ tessellation. Let $P^{*}$ be the fundamental polygon (equilateral triangle) of the $\{3,6\}$ tessellation. Let $l^{*}$ be the edge length of $P^{*}$, which is $l^*=2a$, as shown in Figure~\ref{fig5}. Thus, the area of the triangle, $A_{P^{*}}$, is
\begin{eqnarray} \label{areaT}
A_{P^{*}}= \dfrac{3 \sqrt{3} l^2}{4},
\end{eqnarray}
and the number of faces of the dual tessellation is an integer
\begin{eqnarray} \label{nfTa}
n_f^*=\dfrac{A_{P^{\prime }}}{A_{P^{*}}}=2\xi^2.
\end{eqnarray}
Therefore, the dual code length is $n= n_f^* \dfrac{q}{2} = 2 \xi^2 \dfrac{3}{2}=3 \xi^2$, and the number of logical qubits is $k = 2$.

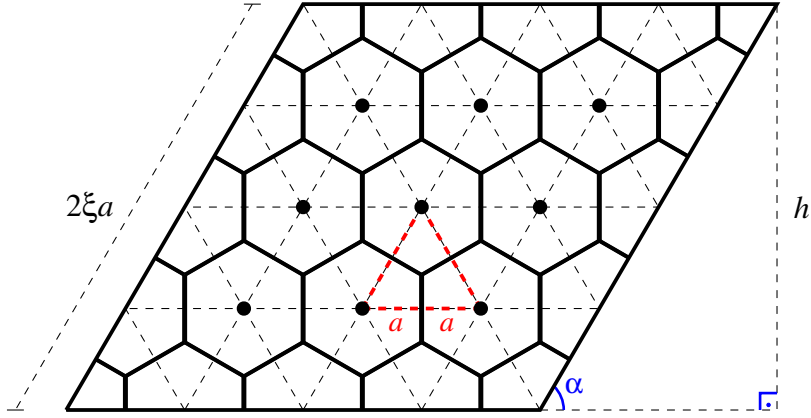
\begin{figure}[h!]
\centering
\begin{tikzpicture}[scale=.45]
\draw[dashed] (3.5,-12) -- (10.5,-12) -- (10.5,0);
\node[right] at (10.7, -6){\large{$h$}};

\begin{scope}[xshift=3.2cm,yshift=-12cm]
\draw[blue,line width=1pt] (0:1cm) arc (0:45:1cm);
\end{scope}
\begin{scope}[xshift=.7cm,yshift=-12.2cm]
\draw[blue,line width=1pt] (9.3,0.2) -- (9.3,.7) -- (9.8,.7);
\fill[blue] (9.55,0.35) circle(0.05);
\draw[blue] (9.55,0.35) circle(0.05);
\end{scope}
\node[blue] at (4.5, -11.3){$\alpha$};

\draw[dashed] (-5,0) -- (-12,-12);
\draw (-4.75,0) -- (-5.25,0);
\draw (-12.25,-12) -- (-11.75,-12);
\node at (-9.8, -6){\large{$2\xi a$}};

\draw[line width=.05cm] (-10.5,-12) -- (3.5,-12) -- (10.5,0) -- (-3.5,0) -- (-10.5,-12);
\clip (-10.5,-12) -- (3.5,-12) -- (10.5,0) -- (-3.5,0) -- (-10.5,-12);

\draw[dashed] (-10.5,-9) -- (7,-9);
\draw[dashed] (-10.5,-6) -- (7,-6);
\draw[dashed] (-10.5,-3) -- (10,-3);
\draw[dashed] (-7,-12) -- (0,0);
\draw[dashed] (-3.5,-12) -- (3.5,0);
\draw[dashed] (0,-12) -- (7,0);
\draw[dashed] (14,-12) -- (7,0);
\draw[dashed] (10.5,-12) -- (3.5,0);
\draw[dashed] (7,-12) -- (0,0);
\draw[dashed] (3.5,-12) -- (-3.5,0);
\draw[dashed] (0,-12) -- (-7,0);
\draw[dashed] (-3.5,-12) -- (-10.5,0);
\draw[dashed] (-7,-12) -- (-14,0);

\foreach \x in {30,90,...,330}
{\draw[line width=.05cm] (\x:2) -- (60+\x:2);}
\begin{scope}[xshift=-3.5cm,yshift=0cm]
\foreach \x in {30,90,...,330}
{\draw[line width=.05cm] (\x:2) -- (60+\x:2);}
\end{scope}
\begin{scope}[xshift=3.5cm,yshift=0cm]
\foreach \x in {30,90,...,330}
{\draw[line width=.05cm] (\x:2) -- (60+\x:2);}
\end{scope}
\begin{scope}[xshift=7cm,yshift=0cm]
\foreach \x in {30,90,...,330}
{\draw[line width=.05cm] (\x:2) -- (60+\x:2);}
\end{scope}
\begin{scope}[xshift=10.5cm,yshift=0cm]
\foreach \x in {30,90,...,330}
{\draw[line width=.05cm] (\x:2) -- (60+\x:2);}
\end{scope}

\begin{scope}[xshift=-5.25cm,yshift=-3cm]
\foreach \x in {30,90,...,330}
{\draw[line width=.05cm] (\x:2) -- (60+\x:2);}
\end{scope}
\begin{scope}[xshift=-1.75cm,yshift=-3cm]
\foreach \x in {30,90,...,330}
{\draw[line width=.05cm] (\x:2) -- (60+\x:2);}
\end{scope}
\begin{scope}[xshift=1.75cm,yshift=-3cm]
\foreach \x in {30,90,...,330}
{\draw[line width=.05cm] (\x:2) -- (60+\x:2);}
\end{scope}
\begin{scope}[xshift=5.25cm,yshift=-3cm]
\foreach \x in {30,90,...,330}
{\draw[line width=.05cm] (\x:2) -- (60+\x:2);}
\end{scope}
\begin{scope}[xshift=8.75cm,yshift=-3cm]
\foreach \x in {30,90,...,330}
{\draw[line width=.05cm] (\x:2) -- (60+\x:2);}
\end{scope}

\begin{scope}[xshift=-3.5cm,yshift=-6cm]
\foreach \x in {30,90,...,330}
{\draw[line width=.05cm] (\x:2) -- (60+\x:2);}
\end{scope}
\begin{scope}[xshift=-7cm,yshift=-6cm]
\foreach \x in {30,90,...,330}
{\draw[line width=.05cm] (\x:2) -- (60+\x:2);}
\end{scope}
\begin{scope}[xshift=0cm,yshift=-6cm]
\foreach \x in {30,90,...,330}
{\draw[line width=.05cm] (\x:2) -- (60+\x:2);}
\end{scope}
\begin{scope}[xshift=3.5cm,yshift=-6cm]
\foreach \x in {30,90,...,330}
{\draw[line width=.05cm] (\x:2) -- (60+\x:2);}
\end{scope}
\begin{scope}[xshift=7cm,yshift=-6cm]
\foreach \x in {30,90,...,330}
{\draw[line width=.05cm] (\x:2) -- (60+\x:2);}
\end{scope}

\begin{scope}[xshift=-8.75cm,yshift=-9cm]
\foreach \x in {30,90,...,330}
{\draw[line width=.05cm] (\x:2) -- (60+\x:2);}
\end{scope}
\begin{scope}[xshift=-5.25cm,yshift=-9cm]
\foreach \x in {30,90,...,330}
{\draw[line width=.05cm] (\x:2) -- (60+\x:2);}
\end{scope}
\begin{scope}[xshift=-1.75cm,yshift=-9cm]
\foreach \x in {30,90,...,330}
{\draw[line width=.05cm] (\x:2) -- (60+\x:2);}
\end{scope}
\begin{scope}[xshift=1.75cm,yshift=-9cm]
\foreach \x in {30,90,...,330}
{\draw[line width=.05cm] (\x:2) -- (60+\x:2);}
\end{scope}
\begin{scope}[xshift=5.25cm,yshift=-9cm]
\foreach \x in {30,90,...,330}
{\draw[line width=.05cm] (\x:2) -- (60+\x:2);}
\end{scope}

\begin{scope}[xshift=-10.5cm,yshift=-12cm]
\foreach \x in {30,90,...,330}
{\draw[line width=.05cm] (\x:2) -- (60+\x:2);}
\end{scope}
\begin{scope}[xshift=-7cm,yshift=-12cm]
\foreach \x in {30,90,...,330}
{\draw[line width=.05cm] (\x:2) -- (60+\x:2);}
\end{scope}
\begin{scope}[xshift=-3.5cm,yshift=-12cm]
\foreach \x in {30,90,...,330}
{\draw[line width=.05cm] (\x:2) -- (60+\x:2);}
\end{scope}
\begin{scope}[xshift=0cm,yshift=-12cm]
\foreach \x in {30,90,...,330}
{\draw[line width=.05cm] (\x:2) -- (60+\x:2);}
\end{scope}
\begin{scope}[xshift=3.5cm,yshift=-12cm]
\foreach \x in {30,90,...,330}
{\draw[line width=.05cm] (\x:2) -- (60+\x:2);}
\end{scope}

\draw[red, dashed, line width=.05cm] (-1.75,-9) -- (1.75,-9);
\draw[red, dashed, line width=.05cm] (-1.75,-9) -- (0,-6);
\draw[red, dashed, line width=.05cm] (1.75,-9) -- (0,-6);
\node[red, below] at (-0.75,-9){$a$};
\node[red, below] at (0.75,-9){$a$};

\draw[line width=.1cm] (1.75,-9) circle (0.1);
\draw[line width=.1cm] (-1.75,-9) circle (0.1);
\draw[line width=.1cm] (-5.25,-9) circle (0.1);

\draw[line width=.1cm] (3.5,-6) circle (0.1);
\draw[line width=.1cm] (0,-6) circle (0.1);
\draw[line width=.1cm] (-3.5,-6) circle (0.1);

\draw[line width=.1cm] (5.25,-3) circle (0.1);
\draw[line width=.1cm] (1.75,-3) circle (0.1);
\draw[line width=.1cm] (-1.75,-3) circle (0.1);

\end{tikzpicture}
\caption{$P^{\prime }$ tiled by the $\{6,3\}$ tessellation and its dual $\{3,6\}$. It is also shown the edge length as a function of the apotheme $a$ of the hexagon.} \label{fig5}
\end{figure}

The distance in the dual tessellation is
\begin{eqnarray} \label{dzTa}
d_z  =  \left \lceil \dfrac{d_e}{l(3,6)} \right \rceil  =  \left \lceil \dfrac{\xi l \sqrt{3}}{2l\sqrt{3}/2} \right \rceil  =  \xi.
\end{eqnarray}

Therefore, the family of toric ATQCs has parameters $[[3 \xi^2, 2, \xi \; / \; \lceil \xi \sqrt{3} \rceil]]$.
\end{proof}

\begin{remark}
Note that, if $L=(2 \xi + 1)a$, then $n_f$ is not an integer, which is a necessary condition for tiling $P^{\prime}$. Therefore, there is no code derived from the $\{6,3\}$ tessellation of the planar torus when the parallelogram has edge length $L = (2 \xi + 1)a$.
\end{remark}

\subsubsection{Toric ATQCs with $\lambda l \times \lambda l$ parallelogram as the fundamental polygon}

\begin{corollary}\label{torichex2}
Let $P^{\prime }$ be the parallelogram tiled by the $\{6,3\}$ tessellation and its dual. Let $L=\lambda l$ be the edge length of the parallelogram which is the torus planar model and $\lambda$ be a positive integer. Then, there exists a family of toric ATQCs with parameters $\left[\left[\lambda^2, 2, \left \lceil \dfrac{\lambda}{\sqrt{3}} \right \rceil \; / \; \lambda \right]\right]$. 
\end{corollary}

\begin{proof}

We split the proof into two cases.
\begin{itemize}

\item Case 1: $\lambda=(2 \xi)$ and $3|\xi$

Let $L = (2\xi)l$, with $\xi \in \mathbb{Z}_+^*$, be the edge length of the parallelogram $P^{\prime }$ tiled by the $\{6,3\}$ tessellation, see Figure~\ref{fig6}. The height $h$ of $P^{\prime }$ is $h=\xi l \sqrt{3}$, for the angle $\beta$ is $\beta = \dfrac{\pi}{3}$. Thus, the area $A_{P^{\prime }}$ of $P^{\prime }$ is
\begin{eqnarray} \label{areaPl}
A_{P^{\prime }}=2 \xi^2 l^2 \sqrt{3}.
\end{eqnarray}

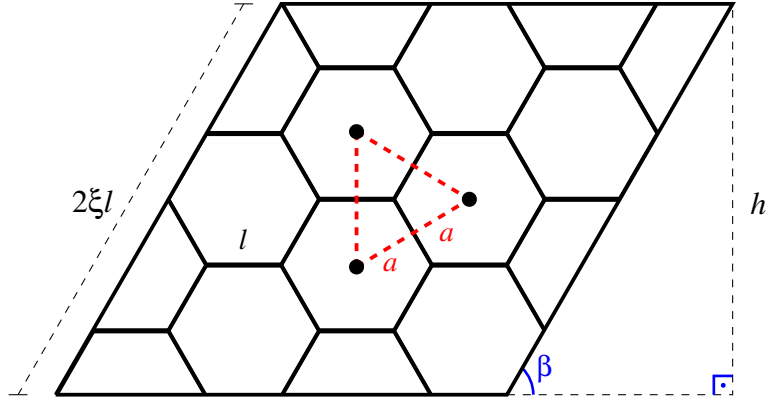
\begin{figure}[h!]
\centering
\begin{tikzpicture}[scale=.5]
\draw[dashed] (1,-12.2) -- (7,-12.2) -- (7,-1.8);
\node[right] at (7.2, -7.1){\large{$h$}};

\begin{scope}[xshift=.7cm,yshift=-12.2cm]
\draw[blue,line width=1pt] (0:1cm) arc (0:45:1cm);
\draw[blue,line width=1pt] (5.8,0) -- (5.8,.5) -- (6.3,.5);
\fill[blue] (6.05,0.25) circle(0.05);
\draw[blue] (6.05,0.25) circle(0.05);
\end{scope}
\node[blue] at (2, -11.5){$\beta$};
\node[above] at (-6,-8.6){$l$};

\draw[dashed] (-6,-1.8) -- (-12,-12.2);
\draw (-5.75,-1.8) -- (-6.25,-1.8);
\draw (-12.25,-12.2) -- (-11.75,-12.2);
\node at (-10, -7.1){\large{$2\xi l$}};
\draw[line width=.05cm] (-11,-12.2) -- (1,-12.2) -- (7,-1.8) -- (-5,-1.8) -- (-11,-12.2);
\clip (-11,-12.2) -- (1,-12.2) -- (7,-1.8) -- (-5,-1.8) -- (-11,-12.2);

\foreach \x in {0,60,...,300}
{\draw[line width=.05cm] (\x:2) -- (60+\x:2);}
\begin{scope}[xshift=-6cm,yshift=0cm]
\foreach \x in {0,60,...,300}
{\draw[line width=.05cm] (\x:2) -- (60+\x:2);}
\end{scope}
\begin{scope}[xshift=0cm,yshift=0cm]
\foreach \x in {0,60,...,300}
{\draw[line width=.05cm] (\x:2) -- (60+\x:2);}
\end{scope}
\begin{scope}[xshift=6cm,yshift=0cm]
\foreach \x in {0,60,...,300}
{\draw[line width=.05cm] (\x:2) -- (60+\x:2);}
\end{scope}

\begin{scope}[xshift=-3cm,yshift=-1.75cm]
\foreach \x in {0,60,...,300}
{\draw[line width=.05cm] (\x:2) -- (60+\x:2);}
\end{scope}
\begin{scope}[xshift=3cm,yshift=-1.75cm]
\foreach \x in {0,60,...,300}
{\draw[line width=.05cm] (\x:2) -- (60+\x:2);}
\end{scope}
\begin{scope}[xshift=9cm,yshift=-1.75cm]
\foreach \x in {0,60,...,300}
{\draw[line width=.05cm] (\x:2) -- (60+\x:2);}
\end{scope}

\begin{scope}[xshift=-6cm,yshift=-3.5cm]
\foreach \x in {0,60,...,300}
{\draw[line width=.05cm] (\x:2) -- (60+\x:2);}
\end{scope}
\begin{scope}[xshift=0cm,yshift=-3.5cm]
\foreach \x in {0,60,...,300}
{\draw[line width=.05cm] (\x:2) -- (60+\x:2);}
\end{scope}
\begin{scope}[xshift=6cm,yshift=-3.5cm]
\foreach \x in {0,60,...,300}
{\draw[line width=.05cm] (\x:2) -- (60+\x:2);}
\end{scope}

\begin{scope}[xshift=-9cm,yshift=-5.25cm]
\foreach \x in {0,60,...,300}
{\draw[line width=.05cm] (\x:2) -- (60+\x:2);}
\end{scope}
\begin{scope}[xshift=-3cm,yshift=-5.25cm]
\foreach \x in {0,60,...,300}
{\draw[line width=.05cm] (\x:2) -- (60+\x:2);}
\end{scope}
\begin{scope}[xshift=3cm,yshift=-5.25cm]
\foreach \x in {0,60,...,300}
{\draw[line width=.05cm] (\x:2) -- (60+\x:2);}
\end{scope}

\begin{scope}[xshift=-6cm,yshift=-7cm]
\foreach \x in {0,60,...,300}
{\draw[line width=.05cm] (\x:2) -- (60+\x:2);}
\end{scope}
\begin{scope}[xshift=0cm,yshift=-7cm]
\foreach \x in {0,60,...,300}
{\draw[line width=.05cm] (\x:2) -- (60+\x:2);}
\end{scope}
\begin{scope}[xshift=6cm,yshift=-7cm]
\foreach \x in {0,60,...,300}
{\draw[line width=.05cm] (\x:2) -- (60+\x:2);}
\end{scope}

\begin{scope}[xshift=-9cm,yshift=-8.75cm]
\foreach \x in {0,60,...,300}
{\draw[line width=.05cm] (\x:2) -- (60+\x:2);}
\end{scope}
\begin{scope}[xshift=-3cm,yshift=-8.75cm]
\foreach \x in {0,60,...,300}
{\draw[line width=.05cm] (\x:2) -- (60+\x:2);}
\end{scope}
\begin{scope}[xshift=3cm,yshift=-8.75cm]
\foreach \x in {0,60,...,300}
{\draw[line width=.05cm] (\x:2) -- (60+\x:2);}
\end{scope}

\begin{scope}[xshift=-6cm,yshift=-10.5cm]
\foreach \x in {0,60,...,300}
{\draw[line width=.05cm] (\x:2) -- (60+\x:2);}
\end{scope}
\begin{scope}[xshift=0cm,yshift=-10.5cm]
\foreach \x in {0,60,...,300}
{\draw[line width=.05cm] (\x:2) -- (60+\x:2);}
\end{scope}
\begin{scope}[xshift=6cm,yshift=-10.5cm]
\foreach \x in {0,60,...,300}
{\draw[line width=.05cm] (\x:2) -- (60+\x:2);}
\end{scope}

\begin{scope}[xshift=-9cm,yshift=-12.25cm]
\foreach \x in {0,60,...,300}
{\draw[line width=.05cm] (\x:2) -- (60+\x:2);}
\end{scope}
\begin{scope}[xshift=-3cm,yshift=-12.25cm]
\foreach \x in {0,60,...,300}
{\draw[line width=.05cm] (\x:2) -- (60+\x:2);}
\end{scope}
\begin{scope}[xshift=3cm,yshift=-12.25cm]
\foreach \x in {0,60,...,300}
{\draw[line width=.05cm] (\x:2) -- (60+\x:2);}
\end{scope}

\draw[red, dashed, line width=.05cm] (-3,-8.8) -- (-3,-5.2);
\draw[red, dashed, line width=.05cm] (-3,-8.8) -- (0,-7);
\draw[red, dashed, line width=.05cm] (-3,-5.2) -- (0,-7);
\node[red, below] at (-2.1,-8.3){$a$};
\node[red, below] at (-0.6,-7.5){$a$};


\draw[line width=.1cm] (-3,-8.8) circle (0.1);
\draw[line width=.1cm] (-3,-5.2) circle (0.1);
\draw[line width=.1cm] (0,-7) circle (0.1);



\end{tikzpicture}
\caption{$P^{\prime }$ tiled by the $\{6,3\}$ tessellation. It is also shown the edge length of $P^{\prime }$ as a function of the edge length $l$ of the hexagon. In particular, $L=2 \xi l$.}
\label{fig6}
\end{figure}

Let $P$ be the fundamental polygon of the $\{6,3\}$ tessellation and $P^{*}$ be the fundamental polygon of the $\{3,6\}$ tessellation. The area $A_{P}$ of the hexagon of edge $l$, and the area $A_{P^{*}}$ of the equilateral triangle of edge $2a$ are given by Eqs.~(\ref{areaH}) and (\ref{areaT}), respectively. Thus, the number of faces of the tessellation and its dual are, respectively, $n_f=\dfrac{A_{P^{\prime }}}{A_P}=\dfrac{4 \xi^2}{3}$ and $n_f^{*}=\dfrac{A_{P^{\prime }}}{A_{P^{*}}}=\dfrac{8 \xi^2}{3}$, where $3|2\xi$ in both cases.

Consequently, the code length in the tessellation is $n=n_f \dfrac{6}{2}=4 \xi^2$, whereas in its dual is $n=n_f^* \dfrac{3}{2}=4 \xi^2$. Since $g=1$ it follows that $k=2$.

Let $d_x$ be the code distance relative to the $X$-errors in the $\{6,3\}$ tessellation. Then,
\begin{eqnarray} \label{dxHl}
d_x = \left \lceil \dfrac{d_e}{l(6,3)} \right \rceil = \left \lceil \dfrac{2 \xi l}{l} \right \rceil = 2 \xi.
\end{eqnarray}

Let $d_z$ be the code distance relative to the $Z$-errors in the dual tessellation $\{3,6\}$. Then,
\begin{eqnarray} \label{dzTl}
d_z = \left \lceil \dfrac{d_e}{l(3,6)} \right \rceil = \left \lceil \dfrac{2 \xi l}{2a} \right \rceil
= \left \lceil \dfrac{2 \xi}{\sqrt{3}} \right \rceil.
\end{eqnarray}

Therefore, there exist ATQCs with parameters $\left[\left[4 \xi^2, 2, \left \lceil \dfrac{2 \xi}{\sqrt{3}} \right \rceil \; / \; 2\xi \right]\right]$

\item Case 2: $\lambda=(2 \xi +1)$ and $3|(2\xi+1)$

Let $L=(2 \xi +1) l$. Then, the height of $P^{\prime }$ is $h=(2 \xi +1)l \frac{\sqrt{3}}{2}$, since the angle $\beta=\dfrac{\pi}{3}$. Consequently, the area of $P^{\prime }$ is
\begin{eqnarray} \label{areaPii}
A_{P^{\prime }}=\dfrac{(2 \xi +1)^2 l^2 \sqrt{3}}{2}
\end{eqnarray}

\begin{center}
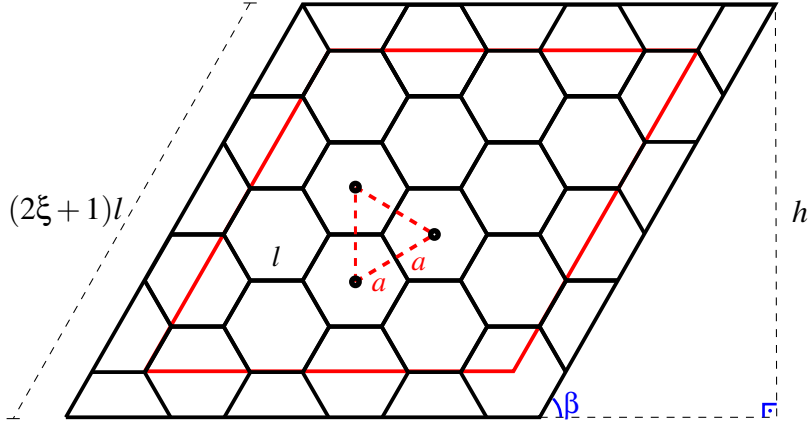
\begin{figure}[h!]
\centering
\begin{tikzpicture}[scale=.35]
\draw[dashed] (4,-13.95) -- (13,-13.95) -- (12.95,1.8);
\node[right] at (13.2, -6.1){\large{$h$}};

\begin{scope}[xshift=3.8cm,yshift=-13.95cm]
\draw[blue,line width=1pt] (0:1cm) arc (0:45:1cm);
\draw[blue,line width=1pt] (8.7,0) -- (8.7,.5) -- (9.2,.5);
\fill[blue] (8.95,0.25) circle(0.05);
\draw[blue] (8.95,0.25) circle(0.05);
\end{scope}
\node[blue] at (5.2, -13.5){$\beta$};
\node[above] at (-6,-8.6){$l$};

\draw[dashed] (-7,1.8) -- (-16,-14);
\draw (-7.25,1.8) -- (-6.75,1.8);
\draw (-16.25,-14) -- (-15.75,-14);
\node at (-14, -6.1){\large{$(2\xi+1) l$}};
\draw[line width=.05cm] (-14,-13.95) -- (4,-13.95) -- (12.95,1.75) -- (-5,1.75) -- (-14,-13.95);
\clip (-14,-13.95) -- (4,-13.95) -- (12.95,1.75) -- (-5,1.75) -- (-14,-13.95);
\draw[line width=.05cm, red] (-11,-12.2) -- (3,-12.2) -- (10,0) -- (-4,0) -- (-11,-12.2);

\foreach \x in {0,60,...,300}
{\draw[line width=.05cm] (\x:2) -- (60+\x:2);}
\begin{scope}[xshift=-6cm,yshift=0cm]
\foreach \x in {0,60,...,300}
{\draw[line width=.05cm] (\x:2) -- (60+\x:2);}
\end{scope}
\begin{scope}[xshift=0cm,yshift=0cm]
\foreach \x in {0,60,...,300}
{\draw[line width=.05cm] (\x:2) -- (60+\x:2);}
\end{scope}
\begin{scope}[xshift=6cm,yshift=0cm]
\foreach \x in {0,60,...,300}
{\draw[line width=.05cm] (\x:2) -- (60+\x:2);}
\end{scope}
\begin{scope}[xshift=12cm,yshift=0cm]
\foreach \x in {0,60,...,300}
{\draw[line width=.05cm] (\x:2) -- (60+\x:2);}
\end{scope}

\begin{scope}[xshift=-3cm,yshift=-1.75cm]
\foreach \x in {0,60,...,300}
{\draw[line width=.05cm] (\x:2) -- (60+\x:2);}
\end{scope}
\begin{scope}[xshift=3cm,yshift=-1.75cm]
\foreach \x in {0,60,...,300}
{\draw[line width=.05cm] (\x:2) -- (60+\x:2);}
\end{scope}
\begin{scope}[xshift=9cm,yshift=-1.75cm]
\foreach \x in {0,60,...,300}
{\draw[line width=.05cm] (\x:2) -- (60+\x:2);}
\end{scope}

\begin{scope}[xshift=-6cm,yshift=-3.5cm]
\foreach \x in {0,60,...,300}
{\draw[line width=.05cm] (\x:2) -- (60+\x:2);}
\end{scope}
\begin{scope}[xshift=0cm,yshift=-3.5cm]
\foreach \x in {0,60,...,300}
{\draw[line width=.05cm] (\x:2) -- (60+\x:2);}
\end{scope}
\begin{scope}[xshift=6cm,yshift=-3.5cm]
\foreach \x in {0,60,...,300}
{\draw[line width=.05cm] (\x:2) -- (60+\x:2);}
\end{scope}

\begin{scope}[xshift=-9cm,yshift=-5.25cm]
\foreach \x in {0,60,...,300}
{\draw[line width=.05cm] (\x:2) -- (60+\x:2);}
\end{scope}
\begin{scope}[xshift=-3cm,yshift=-5.25cm]
\foreach \x in {0,60,...,300}
{\draw[line width=.05cm] (\x:2) -- (60+\x:2);}
\end{scope}
\begin{scope}[xshift=3cm,yshift=-5.25cm]
\foreach \x in {0,60,...,300}
{\draw[line width=.05cm] (\x:2) -- (60+\x:2);}
\end{scope}

\begin{scope}[xshift=-6cm,yshift=-7cm]
\foreach \x in {0,60,...,300}
{\draw[line width=.05cm] (\x:2) -- (60+\x:2);}
\end{scope}
\begin{scope}[xshift=0cm,yshift=-7cm]
\foreach \x in {0,60,...,300}
{\draw[line width=.05cm] (\x:2) -- (60+\x:2);}
\end{scope}
\begin{scope}[xshift=6cm,yshift=-7cm]
\foreach \x in {0,60,...,300}
{\draw[line width=.05cm] (\x:2) -- (60+\x:2);}
\end{scope}

\begin{scope}[xshift=-9cm,yshift=-8.75cm]
\foreach \x in {0,60,...,300}
{\draw[line width=.05cm] (\x:2) -- (60+\x:2);}
\end{scope}
\begin{scope}[xshift=-3cm,yshift=-8.75cm]
\foreach \x in {0,60,...,300}
{\draw[line width=.05cm] (\x:2) -- (60+\x:2);}
\end{scope}
\begin{scope}[xshift=3cm,yshift=-8.75cm]
\foreach \x in {0,60,...,300}
{\draw[line width=.05cm] (\x:2) -- (60+\x:2);}
\end{scope}

\begin{scope}[xshift=-12cm,yshift=-10.5cm]
\foreach \x in {0,60,...,300}
{\draw[line width=.05cm] (\x:2) -- (60+\x:2);}
\end{scope}
\begin{scope}[xshift=-6cm,yshift=-10.5cm]
\foreach \x in {0,60,...,300}
{\draw[line width=.05cm] (\x:2) -- (60+\x:2);}
\end{scope}
\begin{scope}[xshift=0cm,yshift=-10.5cm]
\foreach \x in {0,60,...,300}
{\draw[line width=.05cm] (\x:2) -- (60+\x:2);}
\end{scope}
\begin{scope}[xshift=6cm,yshift=-10.5cm]
\foreach \x in {0,60,...,300}
{\draw[line width=.05cm] (\x:2) -- (60+\x:2);}
\end{scope}

\begin{scope}[xshift=-9cm,yshift=-12.25cm]
\foreach \x in {0,60,...,300}
{\draw[line width=.05cm] (\x:2) -- (60+\x:2);}
\end{scope}
\begin{scope}[xshift=-3cm,yshift=-12.25cm]
\foreach \x in {0,60,...,300}
{\draw[line width=.05cm] (\x:2) -- (60+\x:2);}
\end{scope}
\begin{scope}[xshift=3cm,yshift=-12.25cm]
\foreach \x in {0,60,...,300}
{\draw[line width=.05cm] (\x:2) -- (60+\x:2);}
\end{scope}

\draw[red, dashed, line width=.04cm] (-3,-8.8) -- (-3,-5.2);
\draw[red, dashed, line width=.05cm] (-3,-8.8) -- (0,-7);
\draw[red, dashed, line width=.05cm] (-3,-5.2) -- (0,-7);
\node[red, below] at (-2.1,-8.3){$a$};
\node[red, below] at (-0.6,-7.5){$a$};


\draw[line width=.1cm] (-3,-8.8) circle (0.1);
\draw[line width=.1cm] (-3,-5.2) circle (0.1);
\draw[line width=.1cm] (0,-7) circle (0.1);



\end{tikzpicture}
\caption{$P'$ being tiled by the $\{6, 3\}$ tessellation. It is also shown the edge length of $P'$ as a function
of the edge length $l$ of the hexagon. In particular, $L=(2 \xi +1)l$. The parallelogram corresponding to the red line has $L=(2\xi + 1)l$, however 3 does not divide $2\xi + 1$, so the number of tessellation faces $\{6,3\}$ is not a positive integer.}\label{fig7}
\end{figure}
\end{center}

Since the area $A_{P}$ of the regular hexagon of edge $l$, and the area $A_{P^{*}}$ of the equilateral triangle of edge $2a$ are given by Eqs.~(\ref{areaH}) and (\ref{areaT}), respectively, the number of faces of the tessellation is $n_f= \dfrac{A_{P^{\prime }}}{A_P} = \dfrac{(2 \xi +1)^2}{3}$ and the number of faces of the dual tessellation $\{3,6\}$ is $n_f^{*} = \dfrac{A_{P^{\prime }}}{A_{P^{*}}} = \dfrac{2(2 \xi +1)^2}{3}$. Both $n_f$ and $n_f^{*}$ are integer numbers only if 3 divides $(2 \xi +1)$.

Thus, the code length in the tessellation is $n=n_f \dfrac{6}{2}= (2 \xi +1)^2$ and the code length in the dual tessellation is $n=n_f^* \dfrac{3}{2}=(2 \xi +1)^2$. Again, $k=2$.

The distance $d_x$ in the $\{6,3\}$ tessellation equals
\begin{eqnarray} \label{dxHlii}
d_x = \left \lceil \dfrac{d_e}{l(6,3)} \right \rceil = \left \lceil \dfrac{(2 \xi +1) l}{l} \right \rceil = 2 \xi +1.
\end{eqnarray}

The distance $d_z$ in the $\{3,6\}$ tessellation is given by
\begin{eqnarray} \label{dzTlii}
d_z = \left \lceil \dfrac{d_e}{l(3,6)} \right \rceil = \left \lceil \dfrac{(2 \xi +1) l}{2a} \right \rceil = \left \lceil \dfrac{2 \xi +1}{\sqrt{3}} \right \rceil.
\end{eqnarray}

Therefore, the toric ATQCs have parameters $\left[\left[(2 \xi +1)^2, 2, \left \lceil \frac{2 \xi +1}{\sqrt{3}} \right \rceil \; / \; (2 \xi +1) \right]\right]$.
\end{itemize}
\end{proof}

\begin{remark}
Note that the encoding rate $2/3\xi$ of the codes derived from Corollary~\ref{torichex1} is better than the encoding rate from Case 1, $1/2\xi^2$, and from Case 2, $2/(2 \xi +1)^2$ of the codes derived from Corollary~\ref{torichex2}.
\end{remark}

\subsection{Families of hyperbolic ATQCs}\label{subsec6.2}

In this section families of hyperbolic ATQCs as established in Theorem~\ref{main} are tabulated. Table~\ref{table01} and Table~\ref{table02} illustrate some examples of these families defined in $4g$-gons, with $g \ge 2$.

As seen in Section~\ref{Sec3}, the parameters of the codes are given by
\begin{itemize}
\item Code length - $n = n_{f}\frac{p}{2}$, where $n_f$ is given by Equation (\ref{eqarea});

\item Number of logical qubits - $k = 2g$;

\item Code distances - From Eq.~(\ref{dxdy}), the $X$-distance $d_x$ and the $Z$-distance $d_z$ for the tessellations $\{p,q\}$ and $\{q,p\}$, respectively, as a function of the genus are equal to
\begin{eqnarray} \label{dzdxhyp1}
d_x \geq \left \lceil \frac{d_h}{l (p,q)} \right \rceil = \left \lceil \frac{2 \ \textrm{arccosh} \left(\frac{\cos (\pi/4g)}{\sin (\pi/4g)}\right)}{\textrm{arccosh} \left(\frac{\cos^2 (\pi/q) + \cos (2 \pi/p)}{\sin^2 (\pi/q)}\right)} \right \rceil,
\end{eqnarray}
\begin{eqnarray} \label{dzdxhyp2}
d_z \geq \left \lceil \frac{d_h}{l (q,p)} \right \rceil = \left \lceil \frac{2 \ \textrm{arccosh} \left(\frac{\cos (\pi/4g)}{\sin (\pi/4g)}\right)}{\textrm{arccosh} \left(\frac{\cos^2 (\pi/p) + \cos (2 \pi/q)}{\sin^2 (\pi/p)}\right)} \right \rceil,
\end{eqnarray}
where $d_h$ is given by Eq.~(\ref{eqdh}) and $l(p,q)$ is given by Eq.~(\ref{eql}).
\end{itemize}

Recall that the number of vertices of the tessellation coincides with the number of faces of the dual tessellation $n^{*}_{f}$, then $n^{*}_{f} = n_v = n_f\frac{p}{q}$.

\begin{table*}[ht]
\begin{center}
\caption{Families of Hyperbolic Asymmetric Topological Quantum Codes on $4g$-gon \label{table01}}
\begin{tabular}{|c|c|c|c|c|}\hline \hline
$\{p,q\}$ &   $n_f$                      &  $n_f^*$                         &  $[[n, k, d_z / d_x]]$                      & \mbox{Asymptotic $k/n$}\\ \hline \hline \hline
$\{p,3\}$ & $n_f=\dfrac{12(g-1)}{p-6}$   &  $n_f^*= \dfrac{4p(g-1)}{p-6}$   &  $[[\frac{6p(g-1)}{(p-6)}, 2g, d_z / d_x]]$ & $1/3$ \\ \hline \hline
$\{p,4\}$ & $n_f=\dfrac{8(g-1)}{p-4}$    &  $n_f^*=\dfrac{2p(g-1)}{p-4}$    &  $[[\frac{4p(g-1)}{p-4}, 2g, d_z / d_x]]$   & $1/2$ \\ \hline \hline
$\{p,5\}$ & $n_f=\dfrac{20(g-1)}{3p-10}$ &  $n_f^*=\dfrac{4p(g-1)}{3p-10}$  &  $[[\frac{10p(g-1)}{3p-10}, 2g, d_z /d_x]]$ & $3/5$ \\ \hline \hline
$\{p,6\}$ & $n_f=\dfrac{6(g-1)}{p-3}$    &  $n_f^*=\dfrac{p(g-1)}{p-3}$     &  $[[\frac{3p(g-1)}{p-3}, 2g , d_z / d_x]]$  & $2/3$ \\ \hline
\end{tabular}
\end{center}
\end{table*}

\begin{table*}[ht]
\small{
\begin{center}
\caption{Families of Hyperbolic Asymmetric Topological Quantum Codes on $4g$-gon \label{table02}}
\begin{tabular}{|c|c|c|c|c|} \hline \hline
$\{p,q\}$ &   $n_f$       &    $n_f^*$          &  $[[n, k, d_z / d_x]]$  & $k/n$    \\ \hline \hline \hline
$\{7,3\}$ & $n_f=12(g-1)$ &   $n_f^*=28(g-1)$   &  $\left[\left[42(g-1), 2g, d_z \geq \left \lceil d_h / 1.0906 \right \rceil / d_x \geq \left \lceil d_h / 0.5663 \right \rceil \right]\right]$ & 1/21 \\ \hline \hline
$\{8,3\}$ & $n_f=6(g-1)$  &    $n_f^*=16(g-1)$  &  $\left[\left[24(g-1), 2g, d_z \geq \left \lceil d_h / 1.5286 \right \rceil / d_x \geq \left \lceil d_h / 0.7270 \right \rceil \right]\right]$ & 1/12 \\ \hline \hline
$\{9,3\}$ & $n_f=4(g-1)$  &    $n_f^*=12(g-1)$  &  $\left[\left[18(g-1), 2g, d_z \geq \left \lceil d_h / 1.8551 \right \rceil / d_x \geq \left \lceil d_h / 0.8192 \right \rceil \right]\right]$ & 1/9 \\ \hline \hline
$\{10,3\}$ & $n_f=3(g-1)$ &  $n_f^*=10(g-1)$    &  $\left[\left[15(g-1), 2g, d_z \geq \left \lceil d_h / 2.1226 \right \rceil / d_x \geq \left \lceil d_h / 0.8792 \right \rceil \right]\right]$ & 2/15 \\ \hline \hline
$\{12,3\}$ & $n_f=2(g-1)$ &   $n_f^*=8(g-1)$    &  $\left[\left[12(g-1), 2g, d_z \geq \left \lceil d_h / 2.5534 \right \rceil / d_x \geq \left \lceil d_h / 0.9516 \right \rceil \right]\right]$ & 1/6 \\ \hline \hline \hline
$\{5,4\}$ & $n_f=8(g-1)$ &   $n_f^*=10(g-1)$     &  $\left[\left[20(g-1), 2g, d_z \geq \left \lceil d_h / 1.2537 \right \rceil / d_x \geq \left \lceil d_h / 1.0612 \right \rceil \right]\right]$ & 1/10 \\ \hline \hline
$\{6,4\}$ & $n_f=4(g-1)$ &   $n_f^*=6(g-1)$     &  $\left[\left[12(g-1), 2g, d_z \geq \left \lceil d_h / 1.7628 \right \rceil / d_x \geq \left \lceil d_h / 1.3170 \right \rceil \right]\right]$ & 1/6 \\ \hline \hline
$\{8,4\}$ & $n_f=2(g-1)$ &   $n_f^*=4(g-1)$     &  $\left[\left[8(g-1), 2g, d_z \geq \left \lceil d_h / 2.4485 \right \rceil / d_x \geq \left \lceil d_h / 1.5286 \right \rceil \right]\right]$ & 1/4 \\ \hline \hline \hline
$\{10,5\}$ & $n_f=(g-1)$ &   $n_f^*=2(g-1)$     &  $\left[\left[5(g-1), 2g, d_z \geq \left \lceil d_h / 3.2338 \right \rceil / d_x \geq \left \lceil d_h / 2.1226 \right \rceil \right]\right]$ & 2/5 \\ \hline
\end{tabular}
\end{center}
}
\end{table*}

Figure~\ref{fig33} illustrates the $[[24,4,2 / 5]]$ code, obtained by tiling the fundamental polygon $P^{\prime}$ of the $\{8,8\}$ tessellation with the fundamental polygon $P$ of the $\{8,3\}$ tessellation. Note that the fundamental polygon $P^{\prime}$ identifies the surface $\mathcal{S}_{2}$ of genus $g=2$, and so a $2$-tori. Observe that the number of faces is given by: one face consisting of eight pink and eight white triangles, four faces each consisting of eight yellow and eight white triangles, and one face consisting of eight blue and eight white triangles yielding $n_f = 6$. Hence, $n=n_f.p/2 = 6.8/2=24$.


\begin{table}[h!]
\begin{center}
\caption{Families of Hyperbolic Asymmetric Topological Quantum Codes on $4g$-gon \label{table01}}
\begin{tabular}{|c|c|c|c|c||}\hline \hline
$\{p,q\}$ &   $n_f$                      &  $n_f^*$                         &  $[[n, k, d_z / d_x]]$                      & \mbox{Asymptotic $k/n$}\\ \hline \hline \hline
$\{p,3\}$ & $n_f=\dfrac{12(g-1)}{p-6}$   &  $n_f^*= \dfrac{4p(g-1)}{p-6}$   &  $[[\frac{6p(g-1)}{(p-6)}, 2g, d_z / d_x]]$ & $1/3$ \\ \hline \hline
$\{p,4\}$ & $n_f=\dfrac{8(g-1)}{p-4}$    &  $n_f^*=\dfrac{2p(g-1)}{p-4}$    &  $[[\frac{4p(g-1)}{p-4}, 2g, d_z / d_x]]$   & $1/2$ \\ \hline \hline
$\{p,5\}$ & $n_f=\dfrac{20(g-1)}{3p-10}$ &  $n_f^*=\dfrac{4p(g-1)}{3p-10}$  &  $[[\frac{10p(g-1)}{3p-10}, 2g, d_z /d_x]]$ & $3/5$ \\ \hline \hline
$\{p,6\}$ & $n_f=\dfrac{6(g-1)}{p-3}$    &  $n_f^*=\dfrac{p(g-1)}{p-3}$     &  $[[\frac{3p(g-1)}{p-3}, 2g , d_z / d_x]]$  & $2/3$ \\ \hline
\end{tabular}
\end{center}
\end{table}

\begin{table}[h!]
\small{
\begin{center}
\caption{Families of Hyperbolic Asymmetric Topological Quantum Codes on $4g$-gon \label{table02}}
\begin{tabular}{|c|c|c|c|c|} \hline \hline
$\{p,q\}$ &   $n_f$       &    $n_f^*$          &  $[[n, k, d_z / d_x]]$  & $k/n$    \\ \hline \hline \hline
$\{7,3\}$ & $n_f=12(g-1)$ &   $n_f^*=28(g-1)$   &  $\left[\left[42(g-1), 2g, d_z \geq \left \lceil d_h / 1.0906 \right \rceil / d_x \geq \left \lceil d_h / 0.5663 \right \rceil \right]\right]$ & 1/21 \\ \hline \hline
$\{8,3\}$ & $n_f=6(g-1)$  &    $n_f^*=16(g-1)$  &  $\left[\left[24(g-1), 2g, d_z \geq \left \lceil d_h / 1.5286 \right \rceil / d_x \geq \left \lceil d_h / 0.7270 \right \rceil \right]\right]$ & 1/12 \\ \hline \hline
$\{9,3\}$ & $n_f=4(g-1)$  &    $n_f^*=12(g-1)$  &  $\left[\left[18(g-1), 2g, d_z \geq \left \lceil d_h / 1.8551 \right \rceil / d_x \geq \left \lceil d_h / 0.8192 \right \rceil \right]\right]$ & 1/9 \\ \hline \hline
$\{10,3\}$ & $n_f=3(g-1)$ &  $n_f^*=10(g-1)$    &  $\left[\left[15(g-1), 2g, d_z \geq \left \lceil d_h / 2.1226 \right \rceil / d_x \geq \left \lceil d_h / 0.8792 \right \rceil \right]\right]$ & 2/15 \\ \hline \hline
$\{12,3\}$ & $n_f=2(g-1)$ &   $n_f^*=8(g-1)$    &  $\left[\left[12(g-1), 2g, d_z \geq \left \lceil d_h / 2.5534 \right \rceil / d_x \geq \left \lceil d_h / 0.9516 \right \rceil \right]\right]$ & 1/6 \\ \hline \hline \hline
$\{5,4\}$ & $n_f=8(g-1)$ &   $n_f^*=10(g-1)$     &  $\left[\left[20(g-1), 2g, d_z \geq \left \lceil d_h / 1.2537 \right \rceil / d_x \geq \left \lceil d_h / 1.0612 \right \rceil \right]\right]$ & 1/10 \\ \hline \hline
$\{6,4\}$ & $n_f=4(g-1)$ &   $n_f^*=6(g-1)$     &  $\left[\left[12(g-1), 2g, d_z \geq \left \lceil d_h / 1.7628 \right \rceil / d_x \geq \left \lceil d_h / 1.3170 \right \rceil \right]\right]$ & 1/6 \\ \hline \hline
$\{8,4\}$ & $n_f=2(g-1)$ &   $n_f^*=4(g-1)$     &  $\left[\left[8(g-1), 2g, d_z \geq \left \lceil d_h / 2.4485 \right \rceil / d_x \geq \left \lceil d_h / 1.5286 \right \rceil \right]\right]$ & 1/4 \\ \hline \hline \hline
$\{10,5\}$ & $n_f=(g-1)$ &   $n_f^*=2(g-1)$     &  $\left[\left[5(g-1), 2g, d_z \geq \left \lceil d_h / 3.2338 \right \rceil / d_x \geq \left \lceil d_h / 2.1226 \right \rceil \right]\right]$ & 2/5 \\ \hline
\end{tabular}
\end{center}
}
\end{table}

From the definition of the ATQC, there are two possible ways of increasing the number of edges of the tessellation $\{p,q\}$, either by increasing the genus of the compact orientable surface, equivalently, increasing the area of the fundamental region or by selecting a tessellation leading to a greater $n_f$. The consequence of either one of the last two possibilities is that the code length $n$ increases, the encoding rate decreases, and the distances $d_x$ and $d_z$ also increase. Thus, tessellations giving rise to codes with great $d_x$ and $d_z$ distances, lose in the encoding rate.

Figure~\ref{fig88}(a) and (b) illustrate such a behavior between the distances $d_x$ and $d_z$ and the encoding rates $k/n$  versus $g$ of the codes derived from the $\{7,3\}, \{5,4\}$ and $\{10,5\}$ tessellations, respectively. The distances $d_x$ and $d_z$, in general, grow indefinitely and show a faster increase for small values of $g$ and after that, a slowing down increases. Note that the $d_x$ overgrows compared to $d_z$ when considering the $\{7,3\}$ tessellation and its dual, respectively. Note also that the distance $d_x$ of the $\{5,4\}$ tessellation is close to the distance $d_z$ of the $\{3,7\}$ tessellation whereas the distance $d_z$ of the $\{4,5\}$ tessellation remains below the distance $d_z$ of the $\{3,7\}$ tessellation. The distances $d_x$ and $d_z$ of the $\{10,5\}$ and $\{5,10\}$ tessellations are smaller than the previous tessellations; however, its encoding rate is the best among the previous ones.

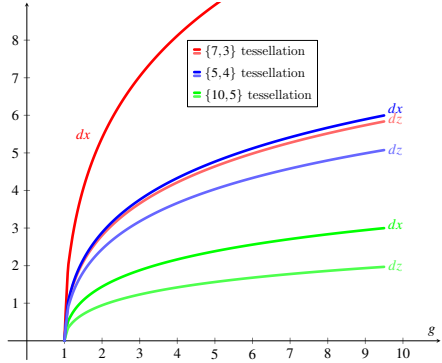
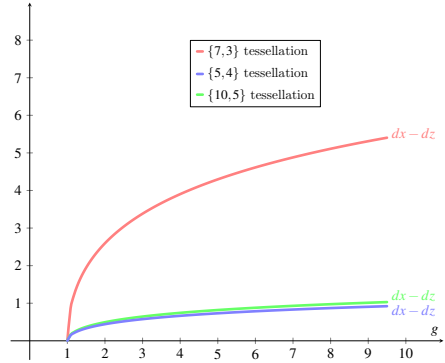
\begin{figure}[h!]
\centering
\subfigure[$d_x$,$d_z$ versus $g$ (without considering the largest integer).]{
\begin{tikzpicture}[scale=0.5,line cap=round,line join=round,>=triangle 45,x=1.0cm,y=1.0cm]

\pgfkeys{/pgf/declare function={acosh(\x) = ln(\x + sqrt(\x^2-1));}}

\begin{axis}[xlabel=$g$, x=1.0cm,y=1.0cm, axis lines=middle, xmin=-.5, xmax=11,
ymin=-.5, ymax=9, xtick={-0,1,...,10}, ytick={-0,1,...,8},]

\draw[line width=2.pt,red,smooth,samples=100,domain=1:9.5] plot(\x,{(2.0*acosh(cos((3.141592653589793/(4.0*\x))*180/pi)/sin((3.141592653589793/(4.0*\x))*180/pi)))/acosh((cos((3.141592653589793/3.0)*180/pi)^(2.0)+cos((2.0*3.141592653589793/7.0)*180/pi))/sin((3.141592653589793/3.0)*180/pi)^(2.0))});
\draw[line width=2.pt,red!60!,smooth,samples=100,domain=1:9.5] plot(\x,{(2.0*acosh(cos((3.141592653589793/(4.0*\x))*180/pi)/sin((3.141592653589793/(4.0*\x))*180/pi)))/acosh((cos((3.141592653589793/7.0)*180/pi)^(2.0)+cos((2.0*3.141592653589793/3.0)*180/pi))/sin((3.141592653589793/7.0)*180/pi)^(2.0))});
\draw[line width=2.pt,green!100!,smooth,samples=100,domain=1:9.5] plot(\x,{(2.0*acosh(cos((3.141592653589793/(4.0*\x))*180/pi)/sin((3.141592653589793/(4.0*\x))*180/pi)))/acosh((cos((3.141592653589793/5.0)*180/pi)^(2.0)+cos((2.0*3.141592653589793/10.0)*180/pi))/sin((3.141592653589793/5.0)*180/pi)^(2.0))});
\draw[line width=2.pt,green!70!,smooth,samples=100,domain=1:9.5] plot(\x,{(2.0*acosh(cos((3.141592653589793/(4.0*\x))*180/pi)/sin((3.141592653589793/(4.0*\x))*180/pi)))/acosh((cos((3.141592653589793/10.0)*180/pi)^(2.0)+cos((2.0*3.141592653589793/5.0)*180/pi))/sin((3.141592653589793/10.0)*180/pi)^(2.0))});
\draw[line width=2.pt,blue,smooth,samples=100,domain=1:9.5] plot(\x,{(2.0*acosh(cos((3.141592653589793/(4.0*\x))*180/pi)/sin((3.141592653589793/(4.0*\x))*180/pi)))/acosh((cos((3.141592653589793/4.0)*180/pi)^(2.0)+cos((2.0*3.141592653589793/5.0)*180/pi))/sin((3.141592653589793/4.0)*180/pi)^(2.0))});
\draw[line width=2.pt,blue!60!,smooth,samples=100,domain=1:9.5] plot(\x,{(2.0*acosh(cos((3.141592653589793/(4.0*\x))*180/pi)/sin((3.141592653589793/(4.0*\x))*180/pi)))/acosh((cos((3.141592653589793/5.0)*180/pi)^(2.0)+cos((2.0*3.141592653589793/4.0)*180/pi))/sin((3.141592653589793/5.0)*180/pi)^(2.0))});

\draw[red] (1.5,5.5) node {$dx$};
\draw[red!60!] (9.8,5.9) node {$dz$};
\draw[blue] (9.8,6.2) node {$dx$};
\draw[blue!60!] (9.8,5.1) node {$dz$};
\draw[green] (9.8,3.1) node {$dx$};
\draw[green!70!] (9.8,2) node {$dz$};

\node [below,text width=3.2cm,draw] at (6,8) {\small{\hspace{0.2cm} $\{7,3\}$ tessellation\\ \hspace{0.2cm} $\{5,4\}$ tessellation \vspace{-0.1cm}\\ \hspace{0.2cm} $\{10,5\}$ tessellation}};
\draw[line width=2pt, red] (4.45,7.75) to (4.6,7.75);
\draw[line width=2pt, red!60!] (4.45,7.65) to (4.6,7.65);
\draw[line width=2pt, blue] (4.45,7.15) to (4.6,7.15);
\draw[line width=2pt, blue!60!] (4.45,7.05) to (4.6,7.05);
\draw[line width=2pt, green] (4.45,6.6) to (4.6,6.6);
\draw[line width=2pt, green!70!] (4.45,6.5) to (4.6,6.5);
\end{axis}
\end{tikzpicture}}
\hspace{0.6cm}
\subfigure[$d_x-d_z$ versus $g$ (without considering the largest integer).]{
\begin{tikzpicture}[scale=0.5,line cap=round,line join=round,>=triangle 45,x=1.0cm,y=1.0cm]

\pgfkeys{/pgf/declare function={acosh(\x) = ln(\x + sqrt(\x^2-1));}}

\begin{axis}[xlabel=$g$, x=1.0cm,y=1.0cm, axis lines=middle, xmin=-.5, xmax=11,
ymin=-.5, ymax=9, xtick={-0,1,...,10}, ytick={-0,1,...,8},]

\draw[line width=2.pt,red!50!,smooth,samples=100,domain=1:9.5] plot(\x,{(2.0*acosh(cos((3.141592653589793/(4.0*\x))*180/pi)/sin((3.141592653589793/(4.0*\x))*180/pi)))/acosh((cos((3.141592653589793/3.0)*180/pi)^(2.0)+cos((2.0*3.141592653589793/7.0)*180/pi))/sin((3.141592653589793/3.0)*180/pi)^(2.0))-(2.0*acosh(cos((3.141592653589793/(4.0*\x))*180/pi)/sin((3.141592653589793/(4.0*\x))*180/pi)))/acosh((cos((3.141592653589793/7.0)*180/pi)^(2.0)+cos((2.0*3.141592653589793/3.0)*180/pi))/sin((3.141592653589793/7.0)*180/pi)^(2.0))});
\draw[line width=2.pt,green!60!,smooth,samples=100,domain=1:9.5] plot(\x,{(2.0*acosh(cos((3.141592653589793/(4.0*\x))*180/pi)/sin((3.141592653589793/(4.0*\x))*180/pi)))/acosh((cos((3.141592653589793/5.0)*180/pi)^(2.0)+cos((2.0*3.141592653589793/10.0)*180/pi))/sin((3.141592653589793/5.0)*180/pi)^(2.0))-(2.0*acosh(cos((3.141592653589793/(4.0*\x))*180/pi)/sin((3.141592653589793/(4.0*\x))*180/pi)))/acosh((cos((3.141592653589793/10.0)*180/pi)^(2.0)+cos((2.0*3.141592653589793/5.0)*180/pi))/sin((3.141592653589793/10.0)*180/pi)^(2.0))});
\draw[line width=2.pt,blue!50!,smooth,samples=100,domain=1:9.5] plot(\x,{(2.0*acosh(cos((3.141592653589793/(4.0*\x))*180/pi)/sin((3.141592653589793/(4.0*\x))*180/pi)))/acosh((cos((3.141592653589793/4.0)*180/pi)^(2.0)+cos((2.0*3.141592653589793/5.0)*180/pi))/sin((3.141592653589793/4.0)*180/pi)^(2.0))-(2.0*acosh(cos((3.141592653589793/(4.0*\x))*180/pi)/sin((3.141592653589793/(4.0*\x))*180/pi)))/acosh((cos((3.141592653589793/5.0)*180/pi)^(2.0)+cos((2.0*3.141592653589793/4.0)*180/pi))/sin((3.141592653589793/5.0)*180/pi)^(2.0))});

\draw[red!50!] (10.2,5.5) node {$dx-dz$};
\draw[blue!50!] (10.2,0.8) node {$dx-dz$};
\draw[green!60!] (10.2,1.2) node {$dx-dz$};

\node [below,text width=3.2cm,draw] at (6,8) {\small{\hspace{0.2cm} $\{7,3\}$ tessellation\\ \hspace{0.2cm} $\{5,4\}$ tessellation \vspace{-0.1cm}\\ \hspace{0.2cm} $\{10,5\}$ tessellation}};
\draw[line width=2pt, red!50!] (4.45,7.7) to (4.6,7.7);
\draw[line width=2pt, blue!50!] (4.45,7.1) to (4.6,7.1);
\draw[line width=2pt, green!60!] (4.45,6.55) to (4.6,6.55);
\end{axis}
\end{tikzpicture}}
\caption{Parameters of the codes generated by the $\{7,3\}$, $\{5,4\}$ and $\{10,5\}$ tessellations.}\label{fig88}
\end{figure}

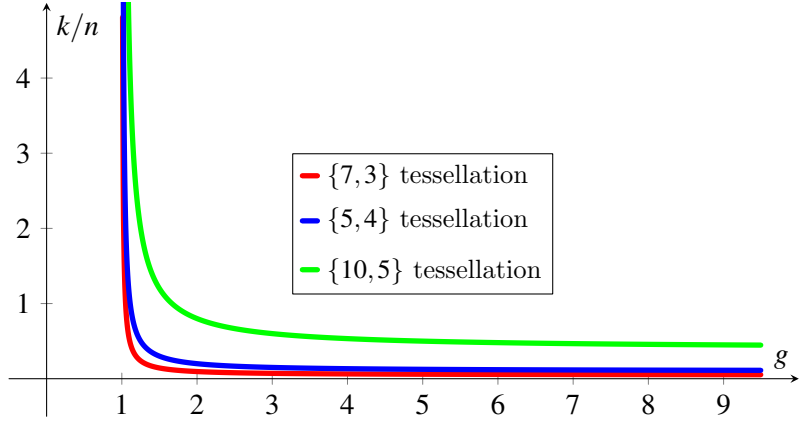
\begin{figure}[h!]
\centering
\begin{tikzpicture}[line cap=round,line join=round,>=triangle 45,x=1.0cm,y=1.0cm]
\begin{axis}[xlabel=$g$, ylabel=$k/n$, x=1.0cm,y=1.0cm, axis lines=middle, xmin=-.5, xmax=10,
ymin=-.5, ymax=5, xtick={0,...,9}, ytick={0,...,4.0},]
\clip(-2,-2) rectangle (10,5);
\draw[line width=2.pt,red,smooth,samples=1000,domain=1.01:9.5] plot(\x,{(2*\x)/(42*(\x-1))});

\draw[line width=2.pt,green,smooth,samples=1000,domain=1.01:9.5] plot(\x,{(2*\x)/(5*(\x-1))});

\draw[line width=2.pt,blue,smooth,samples=1000,domain=1.01:9.5] plot(\x,{(2*\x)/(20*(\x-1))});

\node [below,text width=3.2cm,draw] at (5,3) {\small{\hspace{0.2cm} $\{7,3\}$ tessellation\\ \hspace{0.2cm} $\{5,4\}$ tessellation\\ \hspace{0.2cm} $\{10,5\}$ tessellation}};
\draw[line width=2pt, red] (3.4,2.7) to (3.6,2.7);
\draw[line width=2pt, blue] (3.4,2.1) to (3.6,2.1);
\draw[line width=2pt, green] (3.4,1.45) to (3.6,1.45);
\end{axis}
\end{tikzpicture}
\caption{$k/n$ versus $g$ of the codes generated by the $\{7,3\}, \{5,4\}$ and $\{10,5\}$ tessellations.}\label{fig9}
\end{figure}

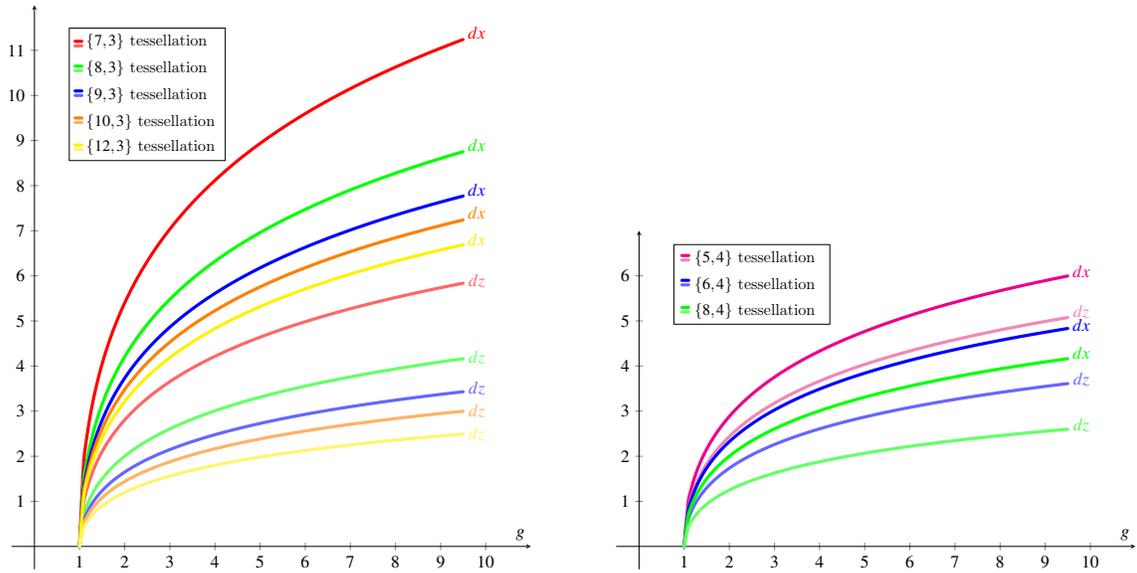
\begin{figure}[h!]
\centering
\subfigure[$d_x$,$d_z$ versus $g$ (without considering the largest integer).]{
\begin{tikzpicture}[scale=0.6,line cap=round,line join=round,>=triangle 45,x=1.0cm,y=1.0cm]

\pgfkeys{/pgf/declare function={acosh(\x) = ln(\x + sqrt(\x^2-1));}}

\begin{axis}[xlabel=$g$, x=1.0cm,y=1.0cm, axis lines=middle, xmin=-.5, xmax=11,
ymin=-.5, ymax=12, xtick={-0,1,...,10}, ytick={-0,1,...,11},]

\draw[line width=2.pt,red,smooth,samples=100,domain=1:9.5] plot(\x,{(2.0*acosh(cos((3.141592653589793/(4.0*\x))*180/pi)/sin((3.141592653589793/(4.0*\x))*180/pi)))/acosh((cos((3.141592653589793/3.0)*180/pi)^(2.0)+cos((2.0*3.141592653589793/7.0)*180/pi))/sin((3.141592653589793/3.0)*180/pi)^(2.0))});
\draw[line width=2.pt,red!60!,smooth,samples=100,domain=1:9.5] plot(\x,{(2.0*acosh(cos((3.141592653589793/(4.0*\x))*180/pi)/sin((3.141592653589793/(4.0*\x))*180/pi)))/acosh((cos((3.141592653589793/7.0)*180/pi)^(2.0)+cos((2.0*3.141592653589793/3.0)*180/pi))/sin((3.141592653589793/7.0)*180/pi)^(2.0))});

\draw[line width=2.pt,green!100!,smooth,samples=100,domain=1:9.5] plot(\x,{(2.0*acosh(cos((3.141592653589793/(4.0*\x))*180/pi)/sin((3.141592653589793/(4.0*\x))*180/pi)))/acosh((cos((3.141592653589793/3.0)*180/pi)^(2.0)+cos((2.0*3.141592653589793/8.0)*180/pi))/sin((3.141592653589793/3.0)*180/pi)^(2.0))});
\draw[line width=2.pt,green!60!,smooth,samples=100,domain=1:9.5] plot(\x,{(2.0*acosh(cos((3.141592653589793/(4.0*\x))*180/pi)/sin((3.141592653589793/(4.0*\x))*180/pi)))/acosh((cos((3.141592653589793/8.0)*180/pi)^(2.0)+cos((2.0*3.141592653589793/3.0)*180/pi))/sin((3.141592653589793/8.0)*180/pi)^(2.0))});

\draw[line width=2.pt,blue,smooth,samples=100,domain=1:9.5] plot(\x,{(2.0*acosh(cos((3.141592653589793/(4.0*\x))*180/pi)/sin((3.141592653589793/(4.0*\x))*180/pi)))/acosh((cos((3.141592653589793/3.0)*180/pi)^(2.0)+cos((2.0*3.141592653589793/9.0)*180/pi))/sin((3.141592653589793/3.0)*180/pi)^(2.0))});
\draw[line width=2.pt,blue!60!,smooth,samples=100,domain=1:9.5] plot(\x,{(2.0*acosh(cos((3.141592653589793/(4.0*\x))*180/pi)/sin((3.141592653589793/(4.0*\x))*180/pi)))/acosh((cos((3.141592653589793/9.0)*180/pi)^(2.0)+cos((2.0*3.141592653589793/3.0)*180/pi))/sin((3.141592653589793/9.0)*180/pi)^(2.0))});

\draw[line width=2.pt,orange,smooth,samples=100,domain=1:9.5] plot(\x,{(2.0*acosh(cos((3.141592653589793/(4.0*\x))*180/pi)/sin((3.141592653589793/(4.0*\x))*180/pi)))/acosh((cos((3.141592653589793/3.0)*180/pi)^(2.0)+cos((2.0*3.141592653589793/10.0)*180/pi))/sin((3.141592653589793/3.0)*180/pi)^(2.0))});
\draw[line width=2.pt,orange!60!,smooth,samples=100,domain=1:9.5] plot(\x,{(2.0*acosh(cos((3.141592653589793/(4.0*\x))*180/pi)/sin((3.141592653589793/(4.0*\x))*180/pi)))/acosh((cos((3.141592653589793/10.0)*180/pi)^(2.0)+cos((2.0*3.141592653589793/3.0)*180/pi))/sin((3.141592653589793/10.0)*180/pi)^(2.0))});

\draw[line width=2.pt,yellow,smooth,samples=100,domain=1:9.5] plot(\x,{(2.0*acosh(cos((3.141592653589793/(4.0*\x))*180/pi)/sin((3.141592653589793/(4.0*\x))*180/pi)))/acosh((cos((3.141592653589793/3.0)*180/pi)^(2.0)+cos((2.0*3.141592653589793/12.0)*180/pi))/sin((3.141592653589793/3.0)*180/pi)^(2.0))});
\draw[line width=2.pt,yellow!70!,smooth,samples=100,domain=1:9.5] plot(\x,{(2.0*acosh(cos((3.141592653589793/(4.0*\x))*180/pi)/sin((3.141592653589793/(4.0*\x))*180/pi)))/acosh((cos((3.141592653589793/12.0)*180/pi)^(2.0)+cos((2.0*3.141592653589793/3.0)*180/pi))/sin((3.141592653589793/12.0)*180/pi)^(2.0))});

\draw[red] (9.8,11.4) node {$dx$};
\draw[red!60!] (9.8,5.9) node {$dz$};
\draw[green] (9.8,8.9) node {$dx$};
\draw[green!60!] (9.8,4.2) node {$dz$};
\draw[blue] (9.8,7.9) node {$dx$};
\draw[blue!60] (9.8,3.5) node {$dz$};
\draw[orange] (9.8,7.4) node {$dx$};
\draw[orange!60!] (9.8,3) node {$dz$};
\draw[yellow] (9.8,6.8) node {$dx$};
\draw[yellow!70!] (9.8,2.5) node {$dz$};

\node [below,text width=3.2cm,draw] at (2.5,11.5) {\small{\hspace{0.25cm}$\{7,3\}$ tessellation\\ \hspace{0.25cm}$\{8,3\}$ tessellation\\ \hspace{0.25cm}$\{9,3\}$ tessellation \\
\hspace{0.25cm}$\{10,3\}$ tessellation \vspace{-0.7cm}\\ \hspace{0.25cm}$\{12,3\}$ tessellation}};
\draw[line width=2pt, red] (0.9,11.2) to (1.05,11.2);
\draw[line width=2pt, red!60!] (0.9,11.1) to (1.05,11.1);
\draw[line width=2pt, green] (0.9,10.65) to (1.05,10.65);
\draw[line width=2pt, green!60!] (0.9,10.55) to (1.05,10.55);
\draw[line width=2pt, blue] (0.9,10.1) to (1.05,10.1);
\draw[line width=2pt, blue!60!] (0.9,10) to (1.05,10);
\draw[line width=2pt, orange] (0.9,9.5) to (1.05,9.5);
\draw[line width=2pt, orange!60!] (0.9,9.4) to (1.05,9.4);
\draw[line width=2pt, yellow] (0.9,8.9) to (1.05,8.9);
\draw[line width=2pt, yellow!70!] (0.9,8.8) to (1.05,8.8);
\end{axis}
\end{tikzpicture}
}
\hspace{0.6cm}
\subfigure[$d_x-d_z$ versus $g$ (without considering the largest integer).]{
\begin{tikzpicture}[scale=0.6,line cap=round,line join=round,>=triangle 45,x=1.0cm,y=1.0cm]

\pgfkeys{/pgf/declare function={acosh(\x) = ln(\x + sqrt(\x^2-1));}}

\begin{axis}[xlabel=$g$, x=1.0cm,y=1.0cm, axis lines=middle, xmin=-.5, xmax=11,
ymin=-.5, ymax=7, xtick={-0,1,...,10}, ytick={-0,1,...,6},]

\draw[line width=2.pt,magenta,smooth,samples=100,domain=1:9.5] plot(\x,{(2.0*acosh(cos((3.141592653589793/(4.0*\x))*180/pi)/sin((3.141592653589793/(4.0*\x))*180/pi)))/acosh((cos((3.141592653589793/4.0)*180/pi)^(2.0)+cos((2.0*3.141592653589793/5.0)*180/pi))/sin((3.141592653589793/4.0)*180/pi)^(2.0))});
\draw[line width=2.pt,magenta!60!,smooth,samples=100,domain=1:9.5] plot(\x,{(2.0*acosh(cos((3.141592653589793/(4.0*\x))*180/pi)/sin((3.141592653589793/(4.0*\x))*180/pi)))/acosh((cos((3.141592653589793/5.0)*180/pi)^(2.0)+cos((2.0*3.141592653589793/4.0)*180/pi))/sin((3.141592653589793/5.0)*180/pi)^(2.0))});

\draw[line width=2.pt,blue,smooth,samples=100,domain=1:9.5] plot(\x,{(2.0*acosh(cos((3.141592653589793/(4.0*\x))*180/pi)/sin((3.141592653589793/(4.0*\x))*180/pi)))/acosh((cos((3.141592653589793/4.0)*180/pi)^(2.0)+cos((2.0*3.141592653589793/6.0)*180/pi))/sin((3.141592653589793/4.0)*180/pi)^(2.0))});
\draw[line width=2.pt,blue!60!,smooth,samples=100,domain=1:9.5] plot(\x,{(2.0*acosh(cos((3.141592653589793/(4.0*\x))*180/pi)/sin((3.141592653589793/(4.0*\x))*180/pi)))/acosh((cos((3.141592653589793/6.0)*180/pi)^(2.0)+cos((2.0*3.141592653589793/4.0)*180/pi))/sin((3.141592653589793/6.0)*180/pi)^(2.0))});

\draw[line width=2.pt,green,smooth,samples=100,domain=1:9.5] plot(\x,{(2.0*acosh(cos((3.141592653589793/(4.0*\x))*180/pi)/sin((3.141592653589793/(4.0*\x))*180/pi)))/acosh((cos((3.141592653589793/4.0)*180/pi)^(2.0)+cos((2.0*3.141592653589793/8.0)*180/pi))/sin((3.141592653589793/4.0)*180/pi)^(2.0))});
\draw[line width=2.pt,green!60!,smooth,samples=100,domain=1:9.5] plot(\x,{(2.0*acosh(cos((3.141592653589793/(4.0*\x))*180/pi)/sin((3.141592653589793/(4.0*\x))*180/pi)))/acosh((cos((3.141592653589793/8.0)*180/pi)^(2.0)+cos((2.0*3.141592653589793/4.0)*180/pi))/sin((3.141592653589793/8.0)*180/pi)^(2.0))});

\draw[magenta] (9.8,6.1) node {$dx$};
\draw[magenta!60!] (9.8,5.2) node {$dz$};
\draw[blue] (9.8,4.9) node {$dx$};
\draw[blue!60!] (9.8,3.7) node {$dz$};
\draw[green] (9.8,4.3) node {$dx$};
\draw[green!70!] (9.8,2.7) node {$dz$};

\node [below,text width=3.2cm,draw] at (2.5,6.7) {\small{\hspace{0.2cm} $\{5,4\}$ tessellation\\ \hspace{0.2cm} $\{6,4\}$ tessellation \vspace{-0.1cm}\\ \hspace{0.2cm} $\{8,4\}$ tessellation}};
\draw[line width=2pt, magenta] (0.95,6.45) to (1.1,6.45);
\draw[line width=2pt, magenta!60!] (0.95,6.35) to (1.1,6.35);
\draw[line width=2pt, blue] (0.95,5.9) to (1.1,5.9);
\draw[line width=2pt, blue!60!] (0.95,5.8) to (1.1,5.8);
\draw[line width=2pt, green] (0.95,5.3) to (1.1,5.3);
\draw[line width=2pt, green!60!] (0.95,5.2) to (1.1,5.2);
\end{axis}
\end{tikzpicture}}
\caption{$d_x$, $d_z$ versus $g$ (without considering the
largest integer) of the codes generated by the family of $\{p,3\}$ and $\{p,4\}$ tessellations.}\label{fig10}
\end{figure}

For some channels, such as the combined amplitude damping and dephasing channel, it is interesting that the difference between the distances $d_z$ and $d_x$ is large. For these channels we can use codes of maximum asymmetry, and in this sense the best option is the codes derived from the $\{7,3\}$ tessellations and its dual, as we can see in Figure~\ref{fig88}(b).

In \cite{ClaPaBra:2014}, families of topological quantum codes derived from the $\{4i + 2,2i + 1\}, \{8i-4,4\}$ and $\{12i-6,3\}$ tessellations, where $i$ is an integer $i \ge 2$, were analyzed; the encoding rates go asymptotically to 1, $1/2$ and $1/3$, respectively. These families of topological quantum codes can be viewed as families of ATQCs by considering the distances $d_x$ and $d_z$ individually, instead of considering the minimum of them. We have noticed that the $d_x$ and $d_z$ distances and its difference, are smaller when considering the $\{4i + 2,2i + 1\}$ tessellations, while the $d_x$ and $d_z$ distances and its difference, are greater when considering the $\{8i-4,4\}$ and $\{12i-6,3\}$ tessellations.

\begin{remark}\label{remark3}
Given an arbitrary surface tiled by a tessellation and its dual, we know that the minimum between the $X$-distance and $Z$-distance can be interchanged: for it is sufficient to interchange the original tessellation with its dual. Therefore, we can always adjust the $Z$-distance being greater than the $X$-distance by properly selecting who will be the tessellation and who will be its dual since the combined amplitude damping and dephasing channel requires this assumption, that is, $d_{z} > d_{x}$.
\end{remark}

\section{Final Remarks} \label{Sec7}

This paper aimed at

\begin{itemize}

\item Establishing the existence and the construction of asymmetric topological quantum code families using the Euclidean and hyperbolic surfaces. These topological quantum codes were constructed from nonself-dual tessellations to guarantee the capacity of correcting asymmetric quantum errors.

\item Introducing the unequal error-protection associated with the nontrivial homological cycle of the $\{p,q\}$ and its dual tessellation,  where $p\neq q$ and $(p-2)(q-2)\ge 4$.

\item Highlighting a family of codes obtained from the $\{7, 3\}$ hyperbolic tessellation whose difference $d_x-d_z$ is very large when compared to the other tessellations due to the fast increasing of $d_x$, as shown in Figure~\ref{fig10}. Note that in this paper we have considered that the face operators are $X$'s and vertices operators are $Z$'s, i.e., we interchange the tessellation for its dual. However, according to Remark~\ref{remark3}, this can be easily performed in order to guarantee that the $Z$-distance is greater than the $X$-distance.

\item Showing an interesting aspect of the class of codes related with the $\{7,3\}$ tessellation is that for any given value of the genus $g$, $d_x \neq d_z$, and therefore, providing the unequal error-protection to the qubits. As an example, this property is not verified for the class of codes derived from the $\{5,4\}$ tessellation for $g$ in the range $2\le g \le 5$ but for $g=4$. When $g=4$, $d_x \neq d_z$, otherwise $d_x = d_z$; equivalently, for the remaining values of $g$ the qubits have equal protection, that is, $(g; d_x / d_z)= \{ (2; 3 / 3), (3; 4 / 4), (5; 5 / 5)\}$. For the classes of codes derived from the $\{6,4\}$ and $\{8,4\}$ when considering $2\le g \le 5$, the unequal error-protection occurs for $2\le g \le 5$ and $3\le g \le 5$, respectively.

\item Showing that the asymptotic encoding rate of the codes resulting from the $\{10,5\}$ tessellation is greater than the asymptotic encoding rate of the codes derived from the $\{7,3\}$ and $\{5,4\}$ tessellations. On the other hand, the performance of the former regarding the unequal error-protection is smaller than the latter, see Figure~\ref{fig9}. As it is well-known in the literature, the code distance increases with either the increase of the fundamental polygon area or selecting a tessellation with a great $n_f$, although the encoding rate decreases also occur with the codes proposed in this paper.

\item Showing an important property which may be extracted from Figure~\ref{fig10}(a). Consider the case with unequal error-protection where $d_x = 6$ and $d_z =3$ for the class of codes from the $\{p,3\}$ family. In particular, consider the cases where $p=7,8,9,10,12$. Note from Figure~\ref{fig10}(a) and by use of the notation $(g; d_x / d_z )$, if we fix $(g; d_x / d_z )= (2; 6 / 3 )$, to achieve the fixed unequal error-protection by the remaining classes of codes the tessellations and the genus are as follows: $\{8,3\}$ tessellation, $(3; 6/3)$; $\{9,3\}$ tessellation, $(4; 6/3)$; $\{10,3\}$ tessellation, $(5; 6/3)$; and $\{12,3\}$ tessellation, $(5; 6/3)$. As can be noticed, the genus increases linearly.

\item Showing another important characteristic which may be extracted from Figure~\ref{fig10}(b). Consider the case with unequal error-protection where $d_x = 5$ and $d_z =4$ for the class of codes from the $\{p,4\}$ family. In particular, consider the cases where $p=5,6,8$. Note from Figure~\ref{fig10}(b) and by use of the notation $(g; d_x / d_z )$, if we fix $(g; d_x / d_z )= (4; 5 / 4 )$, the remaining classes of codes when the unequal error-protection is as fixed, may be achieved with the following values of the tessellation and genus: $\{6,4\}$ tessellation, $(9; 5/4)$; and $\{8,4\}$ tessellation, $(16; 5/4)$. The corresponding code parameters are $[[60,8,5/4]]$, $[[96,18,5/4]]$, and $[[120,32,5/4]]$, respectively. As can be noticed, the genus increases nonlinearly.

\end{itemize}

\section*{Acknowledgement}
This research was supported in part by the Brazilian agency CNPq under Grant 305656/2015-5, 425224/2016-3 and 302759/2017-4.

\small{

}

\end{document}